\documentclass[10pt,journal,compsoc]{IEEEtran}

\ifCLASSOPTIONcompsoc
  \usepackage[nocompress]{cite}
\else
  \usepackage{cite}
\fi

\ifCLASSINFOpdf
\else
\fi

\usepackage[utf8]{inputenc}

\usepackage{cite}
\usepackage{algorithmic}
\usepackage{textcomp}

\usepackage{amssymb}

\usepackage{amsfonts} 

\usepackage{paralist}

\usepackage{amsthm}

\usepackage{blindtext}
\usepackage{mathtools}
\usepackage{booktabs} 
\usepackage{graphicx}
\usepackage{subcaption}
\usepackage{hyperref}
\usepackage{multicol}
\usepackage{longtable}
\usepackage{tabto}
\usepackage{comment}
\usepackage{xspace}
\usepackage{hyperref}
\usepackage{framed}
\usepackage{tcolorbox}
\usepackage{float}

\theoremstyle{definition}

\usepackage{blindtext}
\usepackage[ruled,vlined]{algorithm2e}

\SetAlCapSkip{1em}

\SetKwInput{KwInput}{Input}
\SetKwInput{KwOutput}{Output}

\makeatletter
\newcommand*{\encircle}[1]{\relax\ifmmode\mathpalette\@encircle@math{#1}\else\@encircle{#1}\fi}
\newcommand*{\@encircle@math}[2]{\@encircle{$\m@th#1#2$}}
\newcommand*{\@encircle}[1]{
  \tikz[baseline,anchor=base]{\node[draw,circle,outer sep=0pt,inner sep=.2ex] {#1};}}
\makeatother

\usepackage{array}

\usepackage[thinlines]{easytable}

\usepackage{tikz}
\usepackage{collcell}
\usepackage{rotating}

\usepackage[]{quoting}

\usepackage{enumitem}
\setlistdepth{9}

\newlist{myEnumerate}{enumerate}{9}
\setlist[myEnumerate,1]{label=(\arabic*)}
\setlist[myEnumerate,2]{label=(\Roman*)}
\setlist[myEnumerate,3]{label=(\Alph*)}
\setlist[myEnumerate,4]{label=(\roman*)}
\setlist[myEnumerate,5]{label=(\alph*)}
\setlist[myEnumerate,6]{label=(\arabic*)}
\setlist[myEnumerate,7]{label=(\roman*)}
\setlist[myEnumerate,8]{label=(\alph*)}
\setlist[myEnumerate,9]{label=(\roman*)}

\definecolor{shadecolor}{named}{lightgray}

\usepackage{xcolor}
\hypersetup{
    colorlinks,
    linkcolor={black!50!black},
    urlcolor={blue!80!black}
}
\usepackage{ifthen}
\newboolean{showcomments}
\setboolean{showcomments}{true} 
\ifthenelse{\boolean{showcomments}}
  {\newcommand{\nb}[2]{
    \fcolorbox{red}{yellow}{\bfseries\sffamily\scriptsize#1}
    {\sf\small$\blacktriangleright$\textit{#2}$\blacktriangleleft$}
  }
  
  }
  {\newcommand{\nb}[2]{}
  
  }

\usepackage[normalem]{ulem}

\usepackage{amsthm}

\theoremstyle{definition}
\newtheorem{definition}{Definition}[section]

\theoremstyle{remark}
\newtheorem*{remark}{Remark}

\makeatletter
\newcommand*{\balancecolsandclearpage}{
  \close@column@grid
  \clearpage
  \twocolumngrid
}
\makeatother

\usepackage{amsmath,scalerel}

\makeatletter
\newcommand*\bigcdot{\mathpalette\bigcdot@{.5}}
\newcommand*\bigcdot@[2]{\mathbin{\vcenter{\hbox{\scalebox{#2}{$\m@th#1\bullet$}}}}}
\makeatother

\newtheorem{theorem}{Theorem}[section]
\newtheorem{lemma}[theorem]{Lemma}

\newcommand{\CROME}{{{\sc CROME}}\xspace}
\newcommand{\NAME}{{\sc CROME}\xspace}

\newcommand{\Machine}{M}

\begin{document}

\title{Correct\textendash by\textendash Construction Design of Contextual Robotic Missions Using Contracts}

\author{
  Piergiuseppe~Mallozzi,
  Pierluigi~Nuzzo,
  Nir~Piterman,
  Gerardo~Schneider
  and~Patrizio~Pelliccione
\IEEEcompsocitemizethanks{
  \IEEEcompsocthanksitem P. Mallozzi is at the Department of Electrical Engineering and Computer Science,UC Berkeley. Previously he was with the Department of Computer Science and Engineering, Chalmers University of Technology. E-mail: mallozzi@chalmers.se, mallozzi@berkeley.edu.
  \IEEEcompsocthanksitem P. Nuzzo is at the Ming Hsieh Department of Electrical and Computer Engineering, Univeristy of Southern California, Los Angeles, USA. E-mail: nuzzo@usc.edu.
  \IEEEcompsocthanksitem N. Piterman is at the Department of Computer Science and Engineering, University of Gothenburg. E-mail: piterman@chalmers.se
  \IEEEcompsocthanksitem G. Schneider is at the Department of Computer Science and Engineering, University of Gothenburg. E-mail: gersch@chalmers.se
  \IEEEcompsocthanksitem P. Pelliccione is at Gran Sasso Science Institute (GSSI), L'Aquila, Italy. E-mail: patrizio.pelliccione@gssi.it
  }
}

\IEEEtitleabstractindextext{
\begin{abstract}
  Effectively specifying and implementing robotic missions poses a set of challenges to software engineering for robotic systems. These challenges stem from the need to formalize and execute a robot's high-level tasks while considering various application scenarios and conditions—also known as contexts—in real-world operational environments.

Writing correct mission specifications that explicitly account for multiple contexts can be tedious and error-prone. Furthermore, as the number of contexts—and consequently the complexity of the specification—increases, generating a correct-by-construction implementation (e.g., by using synthesis methods) can become intractable.

A viable approach to address these issues is to decompose the mission specification into smaller, manageable sub-missions, with each sub-mission tailored to a specific context. Nevertheless, this compositional approach introduces its own set of challenges in ensuring the overall mission's correctness.

In this paper, we propose a novel compositional framework for specifying and implementing contextual robotic missions using assume-guarantee contracts. The mission specification is structured in a hierarchical and modular fashion, allowing for each sub-mission to be synthesized as an independent robot controller. We address the problem of dynamically switching between sub-mission controllers while ensuring correctness under predefined conditions.

\end{abstract}

}

\maketitle
\IEEEdisplaynontitleabstractindextext

\section{Introduction}

Robotic missions are increasingly being deployed in various fields, ranging from space exploration to healthcare, search and rescue, and industrial manufacturing. However, the effective specification and implementation of robotic missions remain challenges in software engineering for robotics. We refer to \emph{mission requirements} as a description of the robotic mission in natural language, while \emph{mission specifications} formulate mission requirements in a logical language with precise semantics~\cite{SpecificationPatternsTSE}.

Various approaches have been proposed in recent years to specify missions, for instance, by using logics~\cite{menghi2018multi,ulusoy2011optimal,fainekos2009temporal,guo2013revising,wolff2013automaton,kress2011robot,doi:10.1177/0278364914546174,DBLP:journals/corr/MaozR16,maoz2011aspectltl,maozsynthesis,MaozFSE}, state charts~\cite{bohren2010smach,thomas2013new,klotzbucher2012coordinating}, Petri nets~\cite{wang1991petri,ziparo2008petri}, domain-specific languages~\cite{GLRSWWCA12,campusano.ea:2017:live,DBLP:journals/corr/SchwartzNAM14,Ruscio2016,Bozhinoski2015,Ciccozzi4496,Doherty2012}, or robotic patterns~\cite{SpecificationPatternsTSE,TSE2023}. However, writing mission specifications that reflect the mission requirements is still a challenging and tedious task. This challenge is exacerbated by the variability of conditions and application scenarios that the system may encounter in real-world operational environments. Moreover, as mission specifications become more complex, generating a correct-by-construction implementation, e.g., using synthesis methods, can become intractable. 

One way to address the complexity of mission specifications is to break them down into smaller, more manageable pieces. However, large and complex specifications can be difficult to understand and work with, making it challenging to decompose them into smaller parts. 
Frameworks that can effectively guide the implementation of robotic mission specifications in a modular fashion, via composition of sub-missions, while guaranteeing correctness, are highly desirable.

In this paper, we propose a framework for the specification and implementation of contextual robotic missions. Our approach leverages assume-guarantee contracts to articulate complex missions in a hierarchical and modular manner. These missions are broken down into smaller, more manageable sub-missions, referred to as \emph{tasks}. Central to this methodology is the Contract-based Goal Graph (CGG), an enhanced version of a formal model previously introduced in our earlier work \cite{mallozzi2020crome}. The CGG serves as a pivotal tool for organizing and analyzing the mission's goals, structuring the relationships between tasks in a hierarchical graph.

Each task within this framework is associated with a specific \textit{context}~\cite{mallozzi2020crome}, defined as a set of environmental conditions that may arise during a mission. We define a \textit{mission scenario} as a collection of tasks under compatible contexts. In our model, each scenario is tied to a context that is mutually exclusive to other contexts, ensuring that only one scenario is active at any given moment. Further, each scenario is synthesized into a correct-by-construction finite state machine (i.e., a \textit{controller}) that dictates the robot's responses to the environment's inputs. In this paper, we focus on reactive synthesis from Linear Temporal Logic (LTL), a powerful tool for specifying and verifying the behavior of reactive systems like robots in dynamic environments.

To address the challenge of varying contexts, our framework incorporates a mechanism for dynamic switching between controllers. This feature is crucial for maintaining mission correctness, particularly during \textit{context switches} between different mission scenarios as the environment changes.

Building upon our initial introduction of \NAME in \cite{mallozzi2020crome}, this paper presents significant enhancements and extensions. The original concept of using contracts~\cite{benveniste2018contracts,Nuzzo15b} and context to formalize mission tasks, along with robotic patterns for converting mission requirements into specifications, has been further developed. The key advancements include an improved and more sophisticated version of the Contract-based Goal Graph (CGG), providing a more structured and comprehensive representation of the mission specification. Additionally, we propose a novel methodology for dynamic plan generation. This approach not only ensures the continual satisfaction of the mission specification among context changes but also introduces an efficient synthesis mechanism for associated controllers and a robust system for automatically managing transitions between them.

In this paper, we propose a framework for the specification and implementation of contextual robotic missions. We leverage assume-guarantee contracts to specify a complex mission in a hierarchical and modular way as an 
aggregation of smaller sub-missions, which we term \emph{tasks}. 
We introduce the Contract-based Goal Graph (CGG) as a central component of our approach. The CGG is a formal framework used to hierarchically organize and represent the mission's goals. It serves as a tool to analyze and structure the relationships between tasks.
Each task corresponds to a specific \textit{context}~\cite{mallozzi2020crome}, which is a set of environmental conditions (environment states) that can occur during a mission. We denote by a \textit{mission scenario} a set of tasks that have compatible contexts. All mission scenarios are associated with contexts that are mutually exclusive to other contexts in the mission, so that, at any moment in time, only one context and one scenario are active. Moreover, each scenario can be realized into a controller, which is a finite state machine describing the robot's behavior for every input from the environment. 
In this paper, we focus on reactive synthesis from Linear Temporal Logic (LTL), which is a powerful tool for specifying and verifying the behavior of reactive systems like robots in dynamic environments.
To satisfy a contextual mission, the robot must alternate between mission scenarios as the context changes, leading to a \textit{context switch}. We then propose a method to dynamically switch between task controllers while ensuring the correctness of the overall mission.

Some of the results of this paper have first appeared in our conference submission, where we first introduced 
\NAME~\cite{mallozzi2020crome}, 
where we employed the use of contracts~\cite{benveniste2018contracts,Nuzzo15b} and context to formalize mission tasks, and robotic patterns to capture mission requirements and convert them into specifications. In this paper, we enhance and extend \NAME in two ways. Firstly, we introduce a novel and improved version of the formal model called the \textit{Contract-based Goal Graph} (CGG), which provides a structured representation of the mission specification in graph form. Secondly, we propose a methodology to generate a plan dynamically, which ensures that the mission specification is always satisfied, even when there are context changes. We achieve this by associating controllers with each mission scenario, which can be synthesized efficiently, and providing a mechanism to automatically switch between them while guaranteeing the satisfaction of the overall mission.

\begin{figure}[t]
  \centering
  \includegraphics[width=0.6\columnwidth]{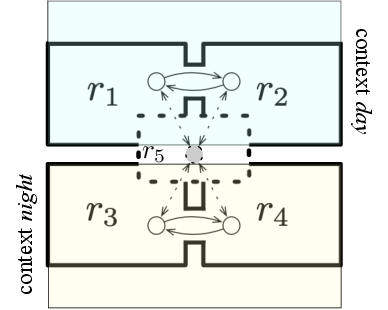}
  \caption{Running example map consisting of five regions and two contexts.}
  \label{fig:illustrative}
 \end{figure}

\subsection{Running Example} \label{sec:running}

We model time as an infinite sequence of discrete steps. At any time step, the robot can be located in one of the five regions in the map shown in Figure~\ref{fig:illustrative}, denoted by $r_1, r_2, r_3, r_4, r_5$, and can move from one region to an adjacent one in a single time step. We assume two mutually exclusive contexts, \textit{day} and \textit{night}. A single sensor is used to detect the presence of a person, denoted by \textit{person}, and two actions are available, \textit{greet} and \textit{register}.

The mission consists of continuously visiting different regions in a context-dependent fashion. During the day, the robot must visit regions $r_1$ and $r_2$ in sequence, starting from $r_1$. During the night, it must visit regions $r_3$ and $r_4$ in order, starting from $r_3$. Whenever the context changes, the robot must resume visiting the regions from the last visited one. When transitioning from regions $r_1, r_2$ to regions $r_3, r_4$, the robot must pass through region $r_5$, which requires an additional time step. Moreover, whenever the robot senses a person in any region, it immediately executes the greeting action. During the day, when a person is detected, the robot must also register them in the next step. Our goal is to generate a controller that can handle any context switch and satisfy the mission.

\subsection{Motivation}

The formalization of mission specifications, such as the one presented in Section~\ref{sec:running}, using a formal language like linear temporal logic (LTL)~\cite{pnueli1977temporal}, is fraught with challenges. These include \textit{1) the need to effectively represent different contexts} and \textit{2) the necessity to maintain the mission's current state across context transitions}.

Correctly modeling the behavior of context variables is crucial for fulfilling contextual missions. If these variables are not properly constrained, the resulting behaviors could lead to either incessant context switching, preventing the mission's completion, or a lack of switching, which could also lead to mission failure. Consider our running example: if the context variable toggles between \textit{day} and \textit{night} at every step, the requirement to pass through region $r_5$ during any context switch would introduce an additional time step. Conversely, specifying minimum or maximum durations for context stability can be cumbersome within the syntax of LTL. Additionally, binding a task to a context may necessitate specialized alterations to the logical formula encoding the task, which can differ depending on the task's nature. For instance, to ensure that the task LTL formulas $\varphi_1 = \mathsf{G}(p_1 \rightarrow p_2)$ and $\varphi_2 = (p_1 ~\mathsf{U} ~p_2)$ are upheld within context $x$, where $p_1$, $p_2$, and $x$ are atomic propositions, it would be necessary to transform the formulas into 
$\varphi_1' = \mathsf{G}((p_1 \land x) \rightarrow (p_2 \land x))$ and 
$\varphi_2' = (((p_1 \land x) \lor \overline{x}) ~\mathsf{U} ~(p_2 \land x))$ to incorporate the context.

Moreover, the robot must keep track of the last visited region when the context shifts. In our scenario, if the context changes while the robot is in $r_1$, upon the return to the \textit{day} context, the robot should resume from the state following $r_1$, namely $r_2$. Such state tracking necessitates additional variables and complex modifications to the LTL task formulas to reflect these variable states.

\subsection{Paper Roadmap}
The structure of this paper is as follows. Section~\ref{sec:background} provides foundational knowledge on assume-guarantee contracts, linear temporal logic, and reactive synthesis. Section~\ref{sec:problem} outlines the problem of achieving a contextual mission. Section~\ref{sec:cgg} introduces the \textit{Contract-Based Goal Graph} (CGG), the formal model created by \NAME to depict contextual missions. In Section~\ref{sec:orchestration}, we explain how \NAME formulates robotic plans that reliably meet mission requirements amidst fluctuating contexts. Details on the implementation and evaluation of \NAME are discussed in Sections~\ref{sec:implementation} and~\ref{sec:evaluation}, respectively. Section~\ref{sec:related} reviews related work in several domains, and the paper concludes with Section~\ref{sec:conclusion}.

\section{Background}
\label{sec:background}

We provide background notions on the basic building blocks of \NAME, namely, contracts, linear temporal logic, and reactive synthesis.

\subsection{Assume-Guarantee Contracts}
\label{sec:backgroundcontracts}

Contract-based design has emerged as a design paradigm that provides formal support for building complex systems in a modular way. This support facilitates compositional reasoning, step-wise refinement of system requirements, and reuse of pre-designed components~\cite{Nuzzo15b, benveniste2018contracts, passerone2019coherent}. 

An \emph{assume-guarantee contract} (or simply \emph{contract}) $\mathcal{C}$ is a triple $(V, A, G)$, where $V$ is a set of variables, including inputs and outputs (or \emph{ports}), and $A$ (assumptions) and $G$ (guarantees) are sets of behaviors over $V$. For simplicity, we often omit $V$ and refer to contracts by their assumptions and guarantees. Assumptions $A$ describe the expected behavior from the environment, while guarantees $G$ specify the system's promised behavior under those assumptions. Behaviors are represented as traces over an alphabet; assumptions and guarantees are sets of traces that satisfy logical formulas. To clarify, we also express assumptions and guarantees using logical formulas.

\subsubsection{Contract Refinement} 

Refinement establishes a pre-order between contracts, formalizing the concept of substitution. Given two contracts $\mathcal{C} = (A, G)$ and $\mathcal{C}' = (A', G')$, we say that $\mathcal{C}$ refines $\mathcal{C}'$, denoted by  $\mathcal{C} \preceq \mathcal{C}'$, if and only if all assumptions of $\mathcal{C}'$ are encompassed by the assumptions of $\mathcal{C}$ and all guarantees of $\mathcal{C}$ are contained within the guarantees of $\mathcal{C}'$; that is, $A \supseteq A' \text{ and } G \subseteq G'$.
Refinement involves weakening the assumptions and strengthening the guarantees. When $\mathcal{C} \preceq \mathcal{C}'$, $\mathcal{C}'$ is considered an \textit{abstraction} of $\mathcal{C}$ and can replace $\mathcal{C}$ in the design.

\subsubsection{Contract Composition}

Contracts specifying different implementations can be combined via the composition operation ($\parallel$). Let $\mathcal{C}_1 = (A_1, G_1)$ and $\mathcal{C}_2 = (A_2, G_2)$ be two contracts. The composition $\mathcal{C}=(A,G)=\mathcal{C}_1  \parallel  \mathcal{C}_2$ can be computed as follows:
\begin{align}
 A & = (A_1 \cap A_2) \cup \overline{(G_1 \cap G_2)}, \label{eq:sat_composition_A}\\
 G & = G_1 \cap G_2. \label{eq:sat_composition_G}
\end{align}
Intuitively, an implementation meeting $\mathcal{C}$ must satisfy both $\mathcal{C}_1$ and $\mathcal{C}_2$'s guarantees, thus the intersection in the second equation.
An environment for $\mathcal{C}$ should also fulfill all assumptions, which explains the conjunction of $A_1$ and $A_2$. However, some assumptions in $\mathcal{C}_1$ may be guaranteed by $\mathcal{C}_2$ and vice versa, allowing the relaxation of $A_1 \cap A_2$ with the complement of $\mathcal{C}$'s guarantees~\cite{benveniste2018contracts}.

\subsubsection{Contract Conjunction}

Different contracts on a single implementation can be combined using the conjunction operation ($\land$). 
For contracts $\mathcal{C}_1 = (A_1, G_1)$ and $\mathcal{C}_2 = (A_2, G_2)$, the conjunction $\mathcal{C}=\mathcal{C}_1  \land  \mathcal{C}_2$ is the most general contract that is more precise than both $\mathcal{C}_1$ and $\mathcal{C}_2$. It is calculated by intersecting the guarantees and uniting the assumptions:

\begin{equation*}
 \mathcal{C} = (A_1 \cup A_2, G_1 \cap G_2).
\end{equation*}
An implementation satisfying \(\mathcal{C}$ must meet both \(\mathcal{C}_1$ and \(\mathcal{C}_2$'s guarantees, while operating under either \(\mathcal{C}_1$ or \(\mathcal{C}_2$'s environment. 

\subsection{Linear Temporal Logic}
\label{sec:ltl}

Given a set of atomic propositions $\mathcal{AP}$, which are Boolean statements over system variables, we define the satisfaction of a proposition $p \in \mathcal{AP}$ by a state $s$ of a system (a specific valuation of the system variables) as $s \models p$, if $p$ is \textit{true} in state $s$.
LTL formulas over $\mathcal{AP}$ are constructed using the following grammar:
\begin{align*}
    \varphi := p ~|~ \neg\varphi ~|~ \varphi_1 \lor \varphi_2 ~|~ \textbf{X} ~\varphi ~|~ \varphi_1 ~\textbf{U}~ \varphi_2,
\end{align*}
where $\varphi$, $\varphi_1$, and $\varphi_2$ are LTL formulas.
Conjunction ($\land$), implication ($\rightarrow$), and equivalence ($\leftrightarrow$) are definable from negation ($\neg$) and disjunction ($\lor$). Boolean constants \textit{true} and \textit{false} are defined as $\textit{true} = \varphi \lor \neg \varphi$ and $\textit{false} = \neg \textit{true}$, respectively. 
The temporal operators $\textbf{X}$ (next) and $\textbf{U}$ (until) allow us to express temporal properties. Additionally, the operators \textit{globally} ($\textbf{G}$) and \textit{eventually} ($\textbf{F}$) are derived as $\textbf{F}~ \varphi = \textit{true} ~\textbf{U}~ \varphi$ and $\textbf{G} ~ \varphi = \neg(\textbf{F}(\neg \varphi))$.
For the formal semantics of LTL, we refer the reader to the literature~\cite{baier2008principles}.

    Both assumptions and guarantees within a contract can be expressed as LTL formulas. In this case, a contract becomes a pair of LTL specifications, one for the assumptions and one for the guarantees.

\subsection{Reactive Synthesis}
\label{sec:reactive_synthesis}

Reactive synthesis is the process of automatically generating a controller, modeled as a finite state machine, from an LTL specification that distinguishes between inputs and outputs. When a controller, satisfying the specification for all possible inputs, is synthesizable, it guarantees the desired system behavior.

Given an LTL formula $\varphi$ over the atomic propositions $\mathcal{AP}$, partitioned into inputs $\mathcal{I}$ and outputs $\mathcal{O}$ such that $\mathcal{AP} = \mathcal{I} \cup \mathcal{O}$, the synthesis problem is to find a finite-state machine $\Machine$, such as a Mealy machine, that \textit{realizes} $\varphi$.

A Mealy machine $\Machine$ is defined as a tuple $(S, \mathcal{I}, \mathcal{O}, s_0, \delta)$, where $S$ is a set of states, $s_0 \in S$ is the initial state, and $\delta: S \times 2^\mathcal{I} \rightarrow S \times 2^\mathcal{O}$ is the transition function. A word $w = (w_0^{\mathcal{I}}, w_0^{\mathcal{O}})(w_1^{\mathcal{I}}, w_1^{\mathcal{O}})(w_2^{\mathcal{I}}, w_2^{\mathcal{O}})... \in (2^{\mathcal{I}} \times 2^{\mathcal{O}})^\omega$ is a \textit{trace} of $\Machine$ if there is a corresponding run $\tau = \tau_0\tau_1\tau_2... \in S^\omega$ with $\tau_0 = s_0$, and for every $i \in \mathbb{N}$, $(\tau_{i+1}, w_i^{\mathcal{O}}) = \delta(\tau_i, w_i^{\mathcal{I}})$. 
$\Machine$ satisfies $\varphi$ if all traces of $\Machine$ satisfy $\varphi$.

Connecting reactive synthesis with contracts, we say that a Mealy machine, or \textit{component}, $\Machine$ is a valid \emph{implementation} of a contract $\mathcal{C}$ if it shares the same variables with $\mathcal{C}$ and all behaviors of $\Machine$ satisfy the guarantees $G$ when constrained by the assumptions $A$. If $\mathcal{C}=(\varphi_A, \varphi_G)$ is a contract of LTL specifications, then we say that $\Machine $satisfies $\mathcal{C} $, denoted $\Machine \models \mathcal{C}$, if $\Machine $satisfies $\varphi_A \rightarrow \varphi_G$. We say that $M$ \textit{realizes} $\mathcal{C}$ if it realizes the LTL formula $\varphi_A \rightarrow \varphi_G$.

    \section{Problem Definition}
    \label{sec:problem}
    
    Robotic missions are designed with the expectation that robots achieve specified objectives within various environmental conditions. These objectives are articulated through \textit{goals}, each underpinned by a \textit{contract}. We regard time as a discrete set of steps, collectively referred to as a \textit{time frame}.
    
    The mission's framework comprises \textit{behaviors} tied to environmental states or \textit{mission contexts}, and the robot's objectives are encapsulated within \textit{goals}.
    
    \begin{definition}[Behavior]
        A \textit{behavior} is an assignment of values to all system (robot and environment) variables at each step $i$ for all $i \in \mathbb{N}$, where $\mathbb{N}$ denotes the set of natural numbers.
    \end{definition}
    
    A context in the mission specification is a dynamic Boolean condition that is evaluated over time. It represents a specific environmental condition that can change as time progresses (for example, "it is raining" or "it is daytime").

    \begin{definition}[Mission Context]
    A \textit{mission context} is a Boolean function over time that characterizes environmental conditions, defined as a mapping $x: \mathbb{N} \rightarrow \{\textit{true}, \textit{false}\}$, with each $x_i$ evaluated at every step to depict the environment's state at that moment.
    \end{definition}
    
    \textbf{Assumptions:}
    \begin{itemize}
    \item The contexts in $X$ are mutually exclusive, ensuring $\forall x_i, x_j \in X$ with $i \neq j$, $x_i$ being \textit{true} at any step implies $x_j$ is \textit{false}.
    \item Each context $x_i \in X$ is \textit{true} infinitely often, which can be denoted by $\mathbf{G}(\mathbf{F}(x_i))$.
    \item A context remains \textit{true} for at least a duration of $t_{\text{context}}$ consecutive steps.
    \end{itemize}
    
    At any point during the mission, one context is \textbf{active}, representing the prevailing environmental condition. While multiple contexts will be active over the mission's course, they never overlap. A \textbf{context switch} is an instantaneous, uncontrollable shift from one active context to another. A fairness assumption is made that each context will activate repeatedly over the course of the mission.
    
    Behaviors are formalized over sets of atomic propositions $\mathcal{I}$ and $\mathcal{O}$, signifying the robot's inputs and outputs, respectively.
    
    \begin{definition}[Robot Behavior]
        A robot behavior $w$ is a sequence $w = (w_0^{\mathcal{I}}, w_0^{\mathcal{O}}), (w_1^{\mathcal{I}}, w_1^{\mathcal{O}}), \ldots$ in $(2^{\mathcal{I}} \times 2^{\mathcal{O}})^\omega$. The pair $w_t = (w_t^{\mathcal{I}}, w_t^{\mathcal{O}})$ represents the input and output at step $t$.
    \end{definition}
    
    \begin{definition}[Projected Robot Behavior]
        The \textit{projected robot behavior} $R(H, w)$ is obtained by filtering the robot behavior $w$ to only include steps in the set $H \subseteq \mathbb{N}$, resulting in a sequence where $t_0, t_1, \ldots \in H$ and $t_0 < t_1 < \ldots$.
    \end{definition}
    
    \begin{definition}[LTL Contract]
        An LTL contract $\mathcal{C}$ is a pair $(\varphi_A, \varphi_G)$ of LTL specifications, with $\varphi_A$ embodying environmental assumptions and $\varphi_G$ specifying the robot's guaranteed behavior when the assumptions $\varphi_A$ are met.
    \end{definition}
    
    \begin{definition}[Goal Context]
        A \textit{goal context} $ctx$ is a static Boolean formula delineating the subset of environmental conditions under which a goal is pertinent.
    \end{definition}

    \begin{definition}[Goal]
        A goal $\mathcal{G}$ is a tuple $(ctx, \mathcal{C})$, where $ctx$ is a goal context signifying the conditions under which the goal is active, and $\mathcal{C}$ is a contract. The set of all goals in a mission is denoted as $G$.
    \end{definition}

    A goal context is a part of the goal's definition and does not change over time. Instead, it is evaluated against the current state of the environment (the active context) to determine if the goal should be pursued at that moment.
    
    A goal context $\text{ctx}$ is \textit{compatible} with a mission context $x_i$ if the static Boolean formula of the goal context and the dynamic Boolean formula of the mission context are satisfiable in conjunction. This compatibility ensures that the goal is relevant and applicable when the specific mission context $x_i$ is active.    
        
    The mission specification, together with the delineation of mission contexts and goal contexts, lays the foundation for producing robot behaviors that fulfill the mission objectives within the constraints of each goal's contract, thereby ensuring that the robot's actions are congruent with the required goals and environmental conditions.

\begin{remark}[Operations on Goals]

    Operations among contracts, such as refinement, conjunction, and composition, are naturally extended to goals, as each goal encompasses a contract. For two goals $\mathcal{G}_1 = (ctx_1, \mathcal{C}_1)$ and $\mathcal{G}_2 = (ctx_2, \mathcal{C}_2)$, with their contracts $\mathcal{C}_1 = (\varphi_{A1}, \varphi_{G1})$ and $\mathcal{C}_2 = (\varphi_{A2}, \varphi_{G2})$, the operations are defined as follows:
    
    \begin{itemize}
        \item \textbf{Refinement:} Goal $\mathcal{G}_1$ refines goal $\mathcal{G}_2$, denoted as $\mathcal{G}_1 \preceq \mathcal{G}_2$, if and only if $\mathcal{C}_1 \preceq \mathcal{C}_2$ (i.e., $\varphi_{A1} \rightarrow \varphi_{A2}$ and $\varphi_{G2} \rightarrow \varphi_{G1}$) and the context $ctx_1$ implies $ctx_2$, formally $ctx_1 \rightarrow ctx_2$.
        
        \item \textbf{Conjunction:} The conjunction of two goals $\mathcal{G}_1 \land \mathcal{G}_2$ results in a goal $\mathcal{G} = (ctx, \mathcal{C})$, where $ctx = ctx_1 \lor ctx_2$ and $\mathcal{C} = \mathcal{C}_1 \land \mathcal{C}_2$ (i.e., $\varphi_A = \varphi_{A1} \lor \varphi_{A2}$ and $\varphi_G = \varphi_{G1} \land \varphi_{G2}$). The combined context $ctx$ represents the union of conditions under which either $\mathcal{G}_1$ or $\mathcal{G}_2$ must hold.
        
        \item \textbf{Composition:} The composition of two goals $\mathcal{G}_1 \parallel \mathcal{G}_2$ yields a goal $\mathcal{G} = (ctx, \mathcal{C})$, where $ctx = ctx_1 \land ctx_2$ and $\mathcal{C} = \mathcal{C}_1 \parallel \mathcal{C}_2$. Here, the intersection of $\varphi_A$ and $\varphi_G$ is taken as in Equations~\ref{eq:sat_composition_A} and~\ref{eq:sat_composition_G}, and the combined context $ctx$ signifies that both conditions $ctx_1$ and $ctx_2$ must be satisfied for $\mathcal{G}$ to be considered active.
    \end{itemize}
    
    In each case, the contexts of the goals are combined according to the logical operations relevant to the contract operation being applied. This ensures that the resulting goal's context and contract are consistent with the semantics of the goal operations.

    \end{remark}

\begin{definition}[Mission]
    A \emph{contextual mission} $\mathcal{M} = (X, G)$, or simply a \emph{mission}, is a pair containing the set of all mission contexts and goals.
    \end{definition}
    
    Given a mission $\mathcal{M} = (X, G)$, we aim to define, for every mission context $x_i \in X$, the set of behaviors that the robot must satisfy when $x_i$ is active. These behaviors are generated by a combination of contracts from goals that are active when $x_i$ is active. We refer to these as \textit{mission scenarios}.
    
    \begin{definition}[Mission Scenario]
    A \emph{mission scenario} is a pair $\mathcal{M}_i = (x_i, \gamma_i)$ where $x_i \in X$ and $\gamma_i$ is an LTL contract.
    \end{definition}
    A mission scenario specifies, for every mutually exclusive context $x_i \in X$, the set of behaviors that must be satisfied by the robot. A mission scenario is \textit{active} when its context is active.

\begin{definition}[Active Signal]
    We define a Boolean condition $\mathcal{A} \in \mathcal{O}$ which we call the \textit{active signal}. This is an output signal controlled by the robot that can be either \textit{true} or \textit{false} at each time step.
\end{definition}

\noindent \textbf{Assumptions:} 
\begin{itemize}
    \item $\mathcal{A}$ becomes \textit{false} at every context-switch and remains \textit{false} for at most $t_{\texttt{trans}}$ time units, where $t_{\texttt{trans}} \ll t_{\texttt{context}}$.
\end{itemize}

\begin{definition}[Indexing Function]
    We define a function $I: X \times \mathcal{O} \rightarrow 2^{\mathbb{N}}$ that takes as input a context $x_i \in X$ and the active signal $\mathcal{A} \in \mathcal{O}$ and returns the set of natural numbers representing the time steps when both $x_i$ and $\mathcal{A}$ are \textit{true}.
\end{definition}

\vspace*{2mm}
Given a mission $\mathcal{M} = \{X, G\}$, we formulate two problems:

\vspace*{2mm}
\noindent{\textbf{Problem 1:}} For each context $x_i \in X$, produce a mission scenario $\mathcal{M}_i = (x_i, \gamma_i)$.

\vspace*{2mm}
\noindent{\textbf{Problem 2:}} For each mission scenario $\mathcal{M}_i = (x_i, \gamma_i)$, produce a robot behavior $w = w_0, w_1, w_2, \ldots$ such that
\begin{equation}
R(I(x_i, \mathcal{A}), w) \models \gamma_i
\end{equation}

Problem 1 seeks to specify the behaviors that the robot must adhere to for each context, whereas Problem 2 aims to generate a robot behavior that satisfies the mission scenario by utilizing the indexing function $I(x_i, \mathcal{A})$ to identify the time steps when the context and active signal are concurrently true.

\begin{figure}[t]
    \centering
    \includegraphics[width=.9\columnwidth]{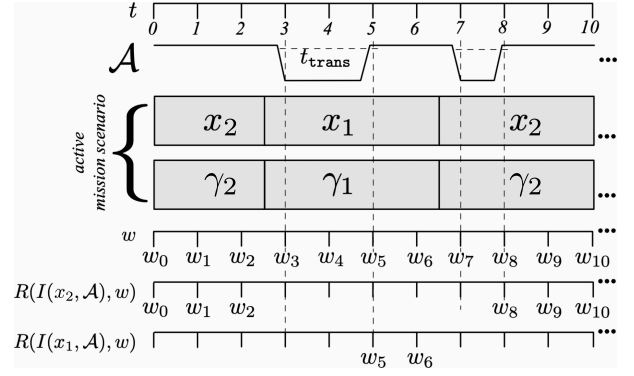}
    \caption{Example of two contextual tasks active at different points in time.}
    \label{fig:timeline-problem}
\end{figure}  
  
\noindent Figure~\ref{fig:timeline-problem} presents an example of active contexts and mission scenarios from $t=0$ to $t=10$, where the overall robot behavior is denoted by $w$ and its projected robot behaviors $R(I(x_1, \mathcal{A}), w)$ and 
$R(I(x_2, \mathcal{A}), w)$ 
are depicted. For the robotic mission to be satisfied, the following conditions must be met:
\begin{itemize}
\item $R(I(x_1, \mathcal{A}), w) = w_5, w_6, \dots \models \gamma_1$
\item 
$R(I(x_2, \mathcal{A}), w) = w_0, w_1, w_2, w_8, w_9, w_{10}, \dots \models \gamma_2$
\end{itemize}

A contextual mission is satisfied if the robot exhibits behavior in accordance with the behaviors $\gamma_i$ of each mission scenario $\mathcal{M}_i = (x_i, \gamma_i)$ when $x_i$ is active. Should a scenario become active after being inactive, the robot is expected to resume from where it left off the last time the scenario was active. In Figure~\ref{fig:timeline-problem}, 
$\mathcal{M}_2=(x_2, \gamma_2)$ 
is the active scenario at step $t=0$. At step $t=3$, the context switches from 
$x_2$ 
to $x_1$, resulting in a different scenario $\mathcal{M}_1$ becoming active. When scenario $\mathcal{M}_3$ becomes active again at step $t=7$, the robot must resume satisfying the behavior $\gamma_3$ from the point it left off at time $t=2$, before the first context switch. We allow for a maximum timeframe $t_{\texttt{trans}}$ following a context switch for the robot to transition to the new scenario.

\begin{remark}[Clarification on Contextual Satisfaction]
Assigning a \texttt{true} value to a goal context $ctx$ does not imply universal satisfaction of the goal contract across the entire timeline. Instead, it indicates that the goal specification aligns with and must be evaluated within all the `\textit{contextual windows}' where the goal has compatible contexts (encompassing all windows in the case of a \texttt{true} context).

Consider a goal $\mathcal{G}=(ctx, \mathcal{C})$. If the context $ctx$ of $\mathcal{G}$ is the Boolean formula \texttt{true}, it signifies that $\mathcal{G}$ is compatible with all contexts. However, this compatibility does not necessitate a \textit{global satisfaction} of the specifications outlined by $\mathcal{C}$ throughout the entire timeline of events.

The fulfillment of $\mathcal{C}$ needs to be assessed by \textit{``stitching together''} segments of the global timeline where $\mathcal{G}$ holds, rather than assuming blanket fulfillment across the entire timeline. This emphasizes a significant distinction from the concept of being \textit{globally true} in LTL, highlighting the contextual nature of goal satisfaction within our framework.
\end{remark}

    The following sections detail how \NAME solves the problems presented. In Section~\ref{sec:cgg}, we address Problem 1, while Section~\ref{sec:orchestration} tackles Problem 2.

 \section{The Contract-Based Goal Graph (CGG)}
 \label{sec:cgg}

 The Contract-Based Goal Graph (CGG) is a structured representation of goals within a robotic mission, organized hierarchically to reflect the relationships and dependencies between different objectives. In this framework, goals are treated as nodes within a graph, and the connections between them are indicative of how they relate to or influence one another. The CGG is specifically designed to facilitate the analysis and synthesis of goal-oriented behavior in complex systems, such as robots operating within dynamic and uncertain environments. See below a detailed explanation.

A CGG is formally defined as a directed graph $T = (\Upsilon, \Xi)$, where each node $\upsilon \in \Upsilon$ represents a goal and each directed edge $\xi \in \Xi$ signifies a relationship between two goals, which can be of three types: refinement, composition, or conjunction. Refinement represents a hierarchical relationship where a goal is broken down into sub-goals. Composition is used to merge goals into a single, more complex goal that captures the behavior necessary to satisfy all individual goals simultaneously. Conjunction links goals that are to be achieved in parallel, indicating that the robot should satisfy all connected goals concurrently, under the assumption that their contexts do not overlap.

The context-based specification clustering algorithm~\cite{mallozzi2020crome} plays a crucial role in addressing Problem 1 (as outlined in Section~\ref{sec:problem}). This algorithm operates as follows:

\begin{enumerate}
 \item \textit{Identify Exclusive Contexts:} The first step is to delineate the mutually exclusive contexts within which the robot operates. These contexts form the basis for clustering goals.
 \item \textit{Cluster Goals by Compatibility:} Goals are then grouped into clusters based on their compatibility with each exclusive context. A goal is considered compatible with a context if the conjunction of the goal's context and the cluster's mutually exclusive context is logically satisfiable.
 \item \textit{Combine Goals Within Clusters:} Within each cluster, goals are combined using the composition operation. This operation generates a new goal that reflects the behaviors needed to satisfy all goals in the cluster, given that the environment satisfies the combined assumptions.
 \item \textit{Connect Clusters:} The new goals produced from the composition within clusters are then connected using the conjunction operation. This step is based on the premise that goals with mutually exclusive contexts can be conjoined because their contexts do not occur simultaneously.
 \item \textit{Derive Mission Scenarios:} Each goal resulting from the conjunction operation corresponds to a unique mission scenario, reflecting a specific set of behaviors that the robot must exhibit when the associated context is active.
\end{enumerate}

This CGG structure ultimately helps in automating the decision-making process for the robot, ensuring that it can dynamically adjust its behavior to fulfill the mission's goals, adapting as different contexts become active or inactive over time. The CGG not only aids in the clear representation of complex missions but also serves as a foundation for algorithms that synthesize robot behaviors that are robust to changes in the environment.

\subsection{Example}

In our running example, we define four goals and two mission contexts, $X=\{\texttt{day}, \texttt{night}\}$. 
For clarity, we provide name, description, context, assumptions, and guarantees for each goal, denoted by $N, D, ctx, A, G$, respectively. We use $\mathcal{G}_i$ and $\mathcal{C}_i$, with $i \in \{1, \ldots, 4\}$, to denote a goal and its contract. For each contract $\mathcal{C}_i$, we use $\psi_i$ and $\phi_i$ to denote the assumption and guarantee formulas, respectively. 

We associate each region on the map with an atomic proposition in  $\{r_1, r_2, r_3, r_4, r_5\}$. When a region is visited, the corresponding atomic proposition is set to \textit{true}. The actions \textit{greet} and \textit{register} are represented by the atomic propositions $g$ and $s$, respectively. We use $p$ to denote the presence of a \textit{person}. The goal contexts $\texttt{day}$ and $\texttt{night}$ are defined such that $\texttt{day} = \overline{\texttt{night}}$ at all times, i.e., $\texttt{day}$ is the negation of $\texttt{night}$.

\begin{align}
  \mathcal{G}_{1} &=
    \begin{cases}
      N & day\_patrolling\\
      D & \text{Keep visiting regions $r_1$ and $r_2$ in order.}\\
      ctx & \texttt{day}\\
      & \mathcal{C}_{1}
      \begin{cases}
        \psi_1 & true\\
        \phi_1 & \textbf{OP}(r_1, r_2)
      \end{cases}\\
    \end{cases} \notag\\
  \mathcal{G}_{2} &=
    \begin{cases}
      N & night\_patrolling\\
      D & \text{Keep visiting regions $r_3$ and $r_4$ in order.}\\
      ctx & \texttt{night}\\
      & \mathcal{C}_{2}
      \begin{cases}
        \psi_2 & true\\
        \phi_2 & \textbf{OP}(r_3, r_4)
      \end{cases}\\
    \end{cases} \notag\\
  \mathcal{G}_{3} &=
    \begin{cases}
      N & day\_register\\
      D & \text{Register a person promptly after being detected.}\\
      ctx & \texttt{day}\\
      & 
\mathcal{C}_{3}
      \begin{cases}
        \psi_3 & \textbf{AE}(p)\\
        \phi_3 & \textbf{BD}(p, s)
      \end{cases}\\
    \end{cases} \notag \\
    \mathcal{G}_{4} &=
    \begin{cases}
      N & always\_greet\\
      D & \text{Greet a person immediately when detected.}\\
      ctx & true\\
      & \mathcal{C}_{4}
      \begin{cases}
        \psi_4 & \textbf{AE}(p)\\
        \phi_4 & \textbf{IR}(p, g)
      \end{cases}\\
    \end{cases} \notag
\end{align}

  \noindent The assumptions $\psi_1$ and $\psi_2$ are set to `true', indicating that these goals do not rely on specific environmental conditions and are always relevant during their respective contexts (\texttt{day} and \texttt{night}). This design choice reflects the ongoing nature of patrolling tasks regardless of additional environmental factors.

\noindent\textbf{AE} (\textbf{AlwaysEventually}) represents the LTL construct globally eventually ($\mathsf{GF}$). \textbf{OP} ($\textbf{OrderedPatrolling}$) is a robotic pattern that expresses the periodic visit of a set of locations in a given order. For instance, $\textbf{OP}(r_1, r_2)$ ensures that regions $r_1$ and $r_2$ are visited repeatedly in that order and it corresponds to the LTL formula:
\begin{align*}
    & \mathsf{G} \mathsf{F} (r_{1} ~~\land~~ \mathsf{F} r_{2}) ~~\land~~ (\overline{r_{2}} ~~\mathbin{\mathsf{U}}~~ r_{1}) ~~\land~~ \\
    & \land~~ \mathsf{G} (r_{2} \rightarrow \mathsf{X} (\overline{r_{2}} ~~\mathbin{\mathsf{U}}~~ r_{1})) ~~\land~~ \mathsf{G} (r_{1} \rightarrow \mathsf{X} (\overline{r_{1}} ~~\mathbin{\mathsf{U}}~~ r_{2}))
\end{align*}

Finally, $\textbf{BD}$ (\textbf{BoundDelay}) and $\textbf{IR}$ (\textbf{InstantaneousReaction}) are robotic patterns which require an action \textit{a} (i.e., setting a proposition to \textit{true}) based on the truth value of an atomic proposition $s$. They correspond to the LTL formulas $\mathsf{G}(s \leftrightarrow a)$ and $\mathsf{G}(s \rightarrow a)$, respectively. For more details and the complete list of robotic patterns, see Menghi et al.~\cite{patterns}. 

\begin{figure}
  \centering
  \includegraphics[width=0.8\columnwidth]{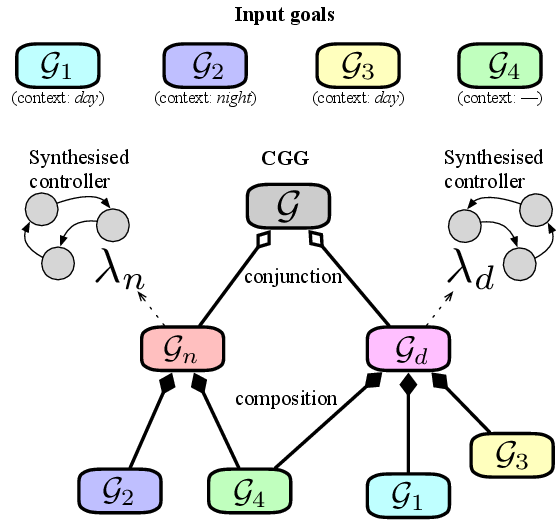}
  \caption{CGG created for the running example.}
  \label{fig:example-cgg-controllers}
 \end{figure}

Given $\mathcal{G}_1$, $\mathcal{G}_2$, $\mathcal{G}_3$, and $\mathcal{G}_4$, the context-based specification clustering algorithm of \NAME produces two mutually exclusive goal contexts, namely, \texttt{day} and \texttt{night}. The context $\texttt{day}$ is associated with the cluster $\{\mathcal{G}_1, \mathcal{G}_3, \mathcal{G}_4\}$, while $\texttt{night}$ is linked to  $\{\mathcal{G}_2, \mathcal{G}_4\}$. $\mathcal{G}_4$ is present in both clusters since its context (\textit{true}) is compatible with any goal context.

  \noindent Figure~\ref{fig:example-cgg-controllers} illustrates the CGG for our example, visually depicting the hierarchical structure and interrelationships between the goals and their contexts. This figure aids in understanding how different goals are clustered and combined based on their contexts.

    \noindent The context $\texttt{day}$ is associated with the cluster $\{\mathcal{G}_1, \mathcal{G}_3, \mathcal{G}_4\}$, while $\texttt{night}$ is linked to  $\{\mathcal{G}_2, \mathcal{G}_4\}$. $\mathcal{G}_4$ appears in both clusters because its context is universally applicable (\textit{true}), making it compatible with any other goal context. This compatibility illustrates how goals can be relevant in multiple mission contexts based on their defined conditions.
    
    \noindent The contracts $\gamma_d$ and $\gamma_n$ result from the composition of the individual goal contracts. They encapsulate the combined requirements for the robot's behavior in both \texttt{day} and \texttt{night} contexts, respectively. This composition ensures that the robot meets all necessary objectives under each context, while respecting the individual constraints of each goal.

    \begin{align}
      \gamma_{d} &=
          \begin{cases}
            \psi_d & \textbf{AE}(p) \lor \overline{\phi_d} \\
            \phi_d & \textbf{OP}(r_1, r_2) \land (\mathsf{GF}(p) \rightarrow \textbf{BD}(p, s))
          \end{cases}\notag
    \end{align}
    \begin{align}
      \gamma_{n} &=
          \begin{cases}
            \psi_n & \textbf{AE}(p) \lor \overline{\phi_n} \\
            \phi_n & \textbf{OP}(r_3, r_4) \land (\textbf{AE}(p) \rightarrow \textbf{IR}(p, g))
          \end{cases}\notag
    \end{align}

        \noindent Controllers $\lambda_d$ and $\lambda_n$ are synthesized to manage the robot's behavior for the \texttt{day} and \texttt{night} contexts, respectively. 
        However, these controllers are context-agnostic. To ensure mission success in a context-sensitive environment, the robot must dynamically switch between these controllers in response to changes in the active context. When the mission context transitions from $\texttt{day}$ to $\texttt{night}$, the robot should switch from $\lambda_d$ to $\lambda_n$ to align its behavior with the new set of active goals. Conversely, when the context reverts to $\texttt{day}$, the robot should switch back to $\lambda_d$, adapting its actions to the reactivated goals of the $\texttt{day}$ context. 
        
        The following sections detail a method to produce such adaptable robot behaviors, ensuring mission satisfaction even in presence of context changes.

\section{Controller Orchestration}
\label{sec:orchestration}

\NAME orchestrates transitions between task controllers to align with the dynamically changing mission scenarios. This orchestration involves task controllers, synthesized for specific mission scenarios, and transition controllers, which facilitate switching between task controllers in response to changes in mission context.

Each controller, modeled as a Mealy machine, operates within designated segments of the robot's state, input, and output space. The orchestration's goal is to maintain continual alignment with the mission's contextual requirements.

\subsection{Task and Transition Controllers}

\begin{figure*}[]
    \centering
    \includegraphics[width=0.6\linewidth]{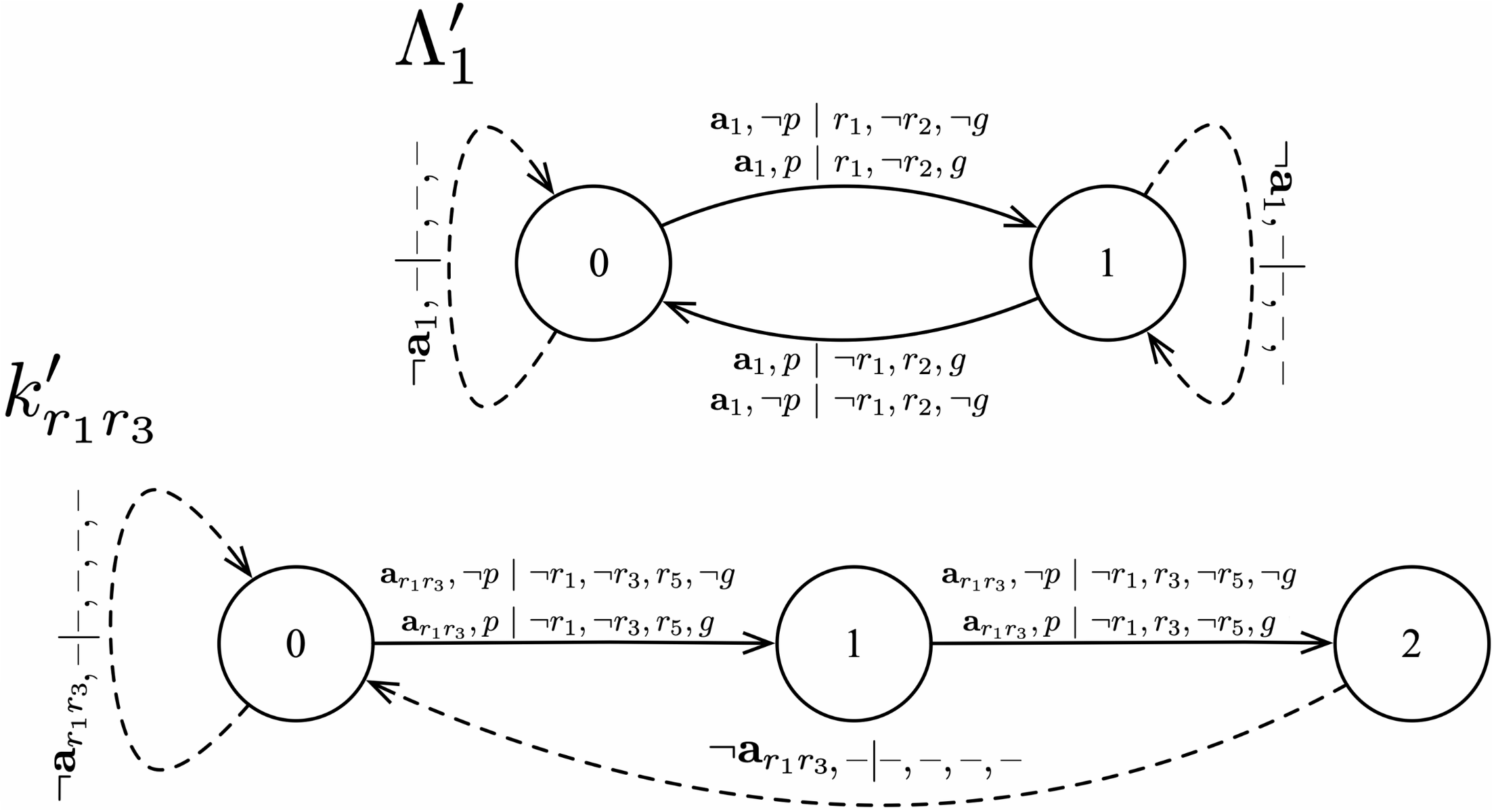}
    \caption{Depiction of task and transition controllers, incorporating activation inputs.} 
    \label{fig:active-controllers}
\end{figure*}
  
  \begin{definition}[Task Controllers]
  For each mission scenario $\mathcal{M}_i = (x_i, \gamma_i)$, a corresponding \textit{task controller} $\lambda_i$ is synthesized to realize the LTL contract $\gamma_i$. The set of all task controllers is denoted as $\Lambda = \{\lambda_1, \lambda_2, \ldots, \lambda_n \}$, where each $\lambda_i$ is tailored to its specific mission scenario.
  \end{definition}

  Transition controllers manage the robot's shift between these mission scenarios.
  
  \begin{definition}[Transition Controllers]
  A \textit{transition controller} $k_{r_i, r_j}$ is a Mealy machine engineered to facilitate efficient movement from region $r_i$ to region $r_j$. These controllers aim to minimize transition time, with $r_i, r_j \in R$ representing distinct regions ($i \neq j$). The complete set of transition controllers, $K$, encompasses all possible transitions, i.e., $K = \{k_{r_ir_j}: r_i, r_j \in R, i \neq j \}$.
  \end{definition}

  For example, in the event of a context switch from day to night, the robot must transition from patrolling regions $\{r_1, r_2\}$ to $\{r_3, r_4\}$. This is managed by the appropriate transition controller.
  
  \begin{definition}[Transition Time]
  The \textit{transition time} $t_{\texttt{trans}}$ is the longest time, in steps, required by any transition controller $k_{r_i, r_j} \in K$ to complete its transition from region $r_i$ to $r_j$.
  \end{definition}

  The orchestration system effectively synchronizes task and transition controllers, ensuring the robot's alignment with the mission scenarios, thereby maintaining continuous satisfaction of the mission's contextual requirements. Further details on this synchronization process are elaborated in the following sections. 
  
\subsubsection{\textbf{Active Controllers}}

\NAME introduces the notion of \textit{active} and \textit{inactive} controllers to modify task and transition controllers and facilitate their coordination. Controllers are finite state machines, specifically Mealy machines.
Two new sets of atomic propositions are defined:
\begin{align*}
    \mathcal{I}_\Lambda^\textbf{a}&=\{\textbf{a}_1, \textbf{a}_2,...,\textbf{a}_n\}\\
    \mathcal{I}_K^\textbf{a} &=\{\textbf{a}_{r_ir_j}: r_i,r_j \in R, i \neq j\}
\end{align*}
These sets of propositions are used as additional inputs for the controllers and are referred to as \textit{activation inputs}.

For a task controller $\lambda_i \in \Lambda$, an atomic proposition $\textbf{a}_i \in \mathcal{I}_\Lambda^\textbf{a}$ is added to its inputs. A self-loop transition is added to any state of $\lambda_i$, which accepts the negation of $\textbf{a}_i$ as input and leaves the output undetermined (\textendash).
For any transition controller $k_{r_ir_j} \in K$, an atomic proposition $\textbf{a}_{r_ir_j}$ is added to its inputs. Two loop transitions are further added, one to its initial state and one from its terminal state to its start state. Both loop transitions accept the negation of $\textbf{a}_{r_ir_j}$ as input for any other input of $k_{r_ir_j}$ and leave the output to be undetermined (\textendash).
The modified task and transition controller sets are denoted by $\Lambda'=\{\lambda_1', \lambda_2',...,\lambda_n'\}$ and $K'=\{k'_{r_ir_j}: r_i,r_j \in R, i \neq j\}$, respectively.

A controller $\lambda_i' \in \Lambda'$ or $k'_{r_ir_j} \in K'$ is considered \textit{active} if its activation input is set to \textit{true}, and \textit{inactive} otherwise. Figure~\ref{fig:active-controllers} illustrates two modified controllers. The task controller $\Lambda_1$ has inputs $\textbf{a}_1$, $p$, outputs $r_1$, $r_2$, $g$, start state $0$, and no terminal state. The transition controller $k_{r_1r_3}$ has inputs $\textbf{a}_{r_1r_3}$, $p$, outputs $r_1$, $r_3$, $r_5$, $g$, start state $0$, and terminal state $3$. \NAME adds the \textit{activation inputs} $\textbf{a}_1$ and $\textbf{a}_{r_1r_3}$ and the transitions represented by dashed lines in Figure~\ref{fig:active-controllers}. 

    We say that two Mealy machines are \textit{trace equivalent} if, for every possible sequence of inputs, they produce the same sequence of outputs. In the context of \NAME, this equivalence implies that the modified controllers ($\lambda_i' \in \Lambda'$ or $k_{r_ir_j}' \in K'$) produce the same behavior as their original counterparts ($\lambda_i \in \Lambda$ or $k_{r_ir_j} \in K$) when the modified controllers are active. Specifically, being \textit{active} means that their respective activation inputs are set to true. This ensures that the original and modified controllers are functionally indistinguishable during their active states, despite the structural modifications introduced by the activation inputs.

    Thus, the traces produced by $\lambda_i' \in \Lambda'$ are equivalent to the traces produced by $\lambda_i \in \Lambda$ when $\lambda_i'$ is \textit{active}. Similarly, the traces produced by $k_{r_ir_j}' \in K'$ are equivalent to those produced by $k_{r_ir_j} \in K$ when $k_{r_ir_j}'$ is \textit{active}. If the activation inputs of $\Lambda'$ and $K'$ are always \textit{true}, then their behaviors are equivalent to those of $\Lambda$ and $K$.
    
\subsubsection{\textbf{Composition}}
Given a set of controllers $\Lambda^*= \Lambda' \cup K'$ where each element $\lambda^i \in \Lambda^*$ is a finite state machine $\lambda^i=(S^i, \mathcal{I}^i, \mathcal{O}^i, s_0^i, \delta^i)$, $N = |\Lambda^*|$, and  where $\textit{in}^i_t$ and $\textit{out}^i_t$ are respectively the inputs and outputs of $\lambda^i$ at time $t$. We indicate with subscripts $s^i_0$, $s^i_t$ and $s^i_{t+1}$ the state of $\lambda^i$ at time $0$, $t$ and $t+1$, respectively. We define the composition of all the elements in $\Lambda^*$ as a finite-state machine $\Lambda^{||}=(S^{||}, \mathcal{I}^{||}, \mathcal{O}^{||}, s_0^{||}, \delta^{||})$ where:
\begin{align*}
    S^{||} &= S^1 \times S^2 \times ... \times S^N\\
    \mathcal{I}^{||} &= \mathcal{I}^1 \times \mathcal{I}^2 \times ... \times \mathcal{I}^N\\
    \mathcal{O}^{||} &= \mathcal{O}^1 \times \mathcal{O}^2 \times ... \times \mathcal{O}^N\\
    s^{||}_0 &= (s_0^1, s_0^2,...,s_0^N)\\
    \delta^{||} &:S^{||} \times \mathcal{I}^{||} \rightarrow S^{||} \times \mathcal{O}^{||}
\end{align*}

At each step, the transition function $\delta^{||}$ produces a new state and output as follows:
\begin{align*}
    &~((s^1_{t+1}, s^2_{t+1}, ..., s^N_{t+1}), (\textit{out}^1_{t}, \textit{out}^2_{t}, ..., \textit{out}^N_{t})) \\
    =&~\delta^{||}((s^1_{t}, s^2_{t}, ..., s^N_{t}), (\textit{in}^1_{t}, \textit{in}^2_{t}, ..., \textit{in}^N_{t}))
\end{align*}
where $(s^i_{t+1}, \textit{out}^i_{t}) = \delta^i(s^i_{t+1}, \textit{in}^i_{t})$ for $i=\{1, \ldots, N\}$. 
The composition is well-defined, i.e., there are no conflicting outputs, 
by construction, since each Mealy machine relates to a different mission scenario, and there is only one active scenario at each time.

\subsection{The Orchestration System}

The orchestration system is responsible for monitoring the active context $x_i \in X$ and coordinating the behavior of the finite state machine $\Lambda^{||}$ to ensure mission fulfillment. This system manipulates the activation inputs within $\mathcal{I}^{||}$, thereby controlling $\Lambda^{||}$'s outputs. If any output $o_t \in \mathcal{O}^{||}$ is undefined by any machine $\lambda_i \in \Lambda^*$ at time $t \in \mathbb{N}$, it defaults to $o_t = \textit{false}$.

We model the orchestration system using a network of timed automata, which are finite state machines extended with clock variables. Our model employs discrete-time ticks. Each component within the orchestration system is represented as a template automaton, as depicted in Figure~\ref{fig:automata}. The transitions in these automata can be labeled with synchronization channels, guards, or updates. Updates may involve simple assignments or complex functions, denoted as $\text{name} := \text{value}$, distinct from predicates ($\text{name} = \text{value}$). State invariants, shown next to states in Figure~\ref{fig:automata}, must be satisfied for the automaton to remain in that state.

\begin{figure}[]
  \centering
  \includegraphics[width=1\linewidth]{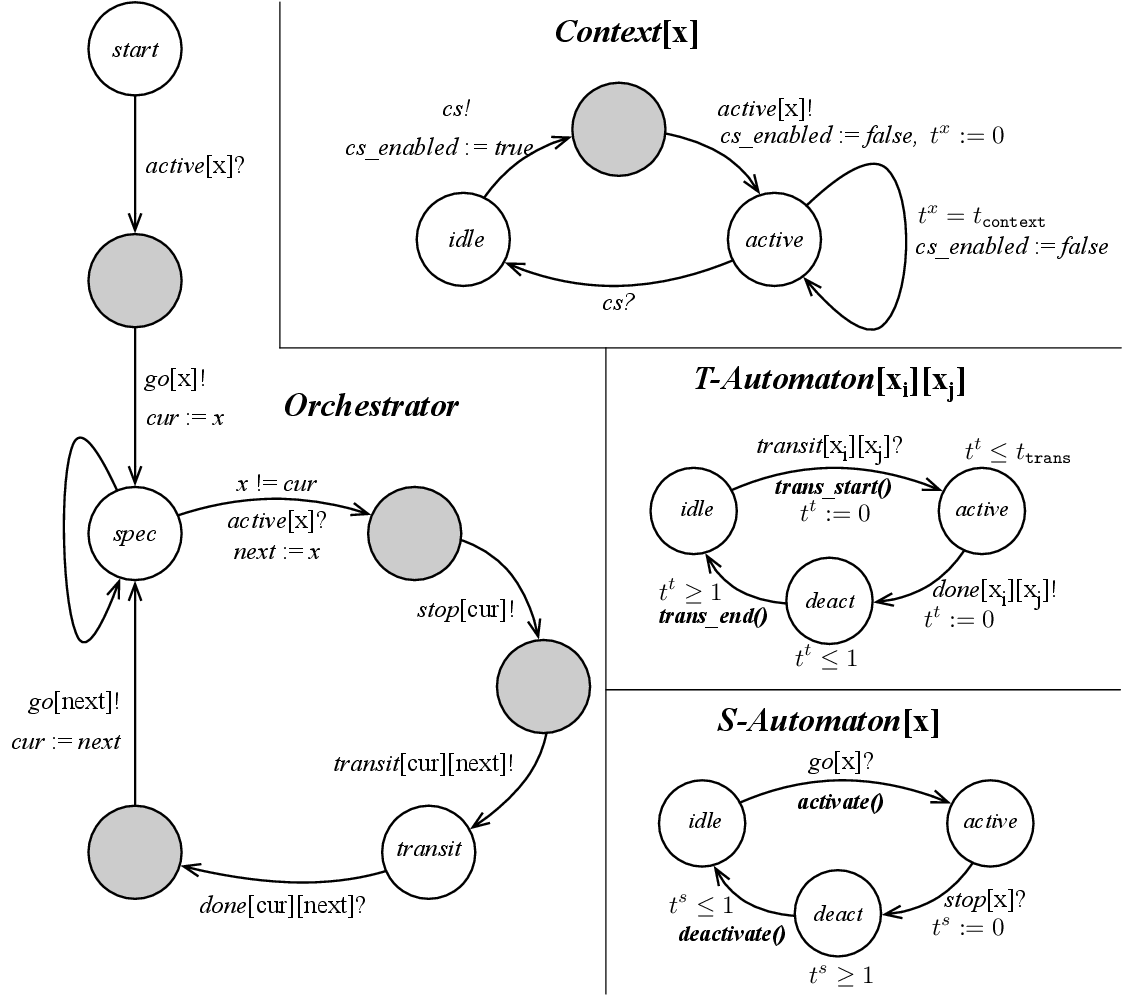}
  \caption{Timed automata modeling interactions among the orchestrator, contexts, and controllers. Gray states indicate \textit{committed locations}.} 
  \label{fig:automata}
\end{figure}

\subsubsection{Template Instances}

There are $n$ instances of the \textit{Context} template, one for each context. Furthermore, we have $(n^2 - n)$ instances of the \textit{T-Automaton} template, each representing a \textit{transition controller} from context $x_i$ to $x_j$ where $i \neq j$. For each \textit{task controller}, there is an associated instance of the \textit{S-Controller} template, with one \textit{S-Controller} for each \textit{Context} instance. The \textit{Orchestrator} automaton coordinates the states of all the \textit{Context}, \textit{S-Controller}, and \textit{T-Automaton} instances. We use the abbreviations $C[i]$, $S[i]$, and $T[i][j]$ for instances of contexts, specification and transition automata, respectively, and $O$ for the orchestration automaton.

\subsubsection{Synchronization Channels}

The automata synchronize via binary synchronization channels. Given $n$ contexts, the following channels are defined:
\begin{itemize}
\item $n$ $active$ channels, one for each context, activated when the corresponding context becomes active.
\item $n$ $go$ and $n$ $stop$ channels, to activate and deactivate the task automata, respectively.
\item $(n^2 - n)$ $transit$ and $(n^2 - n)$ $done$ channels, to start each transition automaton and signal its completion.
\item A \textit{broadcast channel} $cs$, where one sender ($cs!$) synchronizes with multiple receivers ($cs?$). It communicates context switches to all contexts.
\end{itemize}

\subsubsection{States}

\textit{Committed locations}, shown in gray in Figure~\ref{fig:automata}, are states where immediate transitions occur, prohibiting time delays. Thus, transitions from or to committed states are considered atomic.

Initially, all automata are in the \textit{idle} state. At any given moment, context, transition, and task automata can be in either \textit{active} or \textit{idle} states, while the orchestrator can be in the \textit{spec} or \textit{transit} states.

\subsubsection{Local Clock Variables}

Contexts and transition automata each have a local clock variable, $t^x$ and $t^t$ respectively. Contexts use $t^x$ to measure the minimum required active duration, while transition automata use $t^t$ to limit their maximum active time. These clocks are reset when the automata enter their \textit{active} states.

\subsubsection{Synchronization with Mealy Machines}

Timed automata and Mealy machines serve distinct roles in the system. The former models dynamic behavior and time, while the latter determines the robot's actions without specific time constraints. To synchronize these components, we discretize time into fixed units $t$. At each time unit, timed automata transitions occur based on their guards, and concurrently, all Mealy machines in $\Lambda^{||}$ react according to their inputs. Most Mealy machines will be inactive, with only one transition or task controller active at any given time. This activation is governed by the timed automata, which manage the activation inputs in $\mathcal{I}^{||}$.

The functions \textit{\textbf{trans\_start()}}, \textit{\textbf{trans\_end()}}, \textit{\textbf{activate()}}, and \textit{\textbf{deactivate()}} are invoked when transition or task automata move from \textit{idle} to \textit{active}. These functions handle the activation inputs of the Mealy machines in $\Lambda^{||}$, influencing their behavior. Specifically, for any task automaton $S[i]$ and transition automaton $T[i][j]$:
\begin{itemize}
    \item $activate()$ sets the activation input of the corresponding Mealy machine $\lambda_i'$ to true and marks $Active \in \mathcal{O}$ as true.
    \item $deactivate()$ sets the activation input of $\lambda_i'$ to false and marks $Active$ as false.
    \item $trans\_start()$ activates the appropriate Mealy machine $k'_{r_ir_j} \in K'$ based on the robot's current and destination regions.
    \item $trans\_end()$ deactivates the Mealy machine $k'_{r_ir_j} \in K'$.
\end{itemize}

\begin{remark}
\label{rem:active}
$Active \in \mathcal{O}$ is \textit{true} if and only if any task automaton is in the \textit{active} state.
\end{remark}

\subsection{Proof of Correctness}
We prove that our system generates traces that always satisfy the contextual mission as defined in Section~\ref{sec:problem}. This proof is constructed incrementally, focusing on the following properties of the orchestration system:

\begin{itemize}
\item Initially, all automata are in the \textit{idle} state, awaiting an active context.
\item After the first context activation, exactly one context remains active at any given time.
\item A context stays active for at least $t_{\texttt{context}}$ time units.
\item No context can become active in two consecutive time units.
\item Post-initial context activation, either a task or a transition automaton is active, but never both simultaneously.
\item A transition automaton's active period is capped at $t_{\texttt{trans}}$ time units.
\item Following a context activation within $t_{\texttt{trans}}$ time units, the associated task automaton becomes active.
\end{itemize}

We assume all automata start in the \textit{idle} state.

\begin{lemma}
\label{lem:contextfirst}
Initially, only a context automaton can transition out of \textit{idle}, and precisely one context can transition to the \textit{active} state.
\end{lemma}

\begin{proof}
We use induction on the number of contexts.

\vspace{2mm}
\noindent\textbf{Base case $n=1$:} Here, there is one context ($C[0]$), one orchestrator ($O$), and one task automaton ($S[0]$). The orchestrator ($O$) and $S[0]$ are dependent on the synchronization channels $active$ and $go$, respectively. With $cs\_enabled$ initially set to $true$, $C[0]$ can transition from \textit{idle} to \textit{active}, synchronizing on $cs$ and $active[0]$ channels. This transition blocks any other automaton from moving out of \textit{idle}.

\vspace{2mm}
\noindent\textbf{Induction hypothesis:} Assume the lemma holds for $n=k$ contexts.

\vspace{2mm}
\noindent\textbf{Induction step:} For $n=k+1$, all contexts $C[0], C[1],...,C[k]$ are poised to transition from \textit{idle} to \textit{active}. Based on our hypothesis, only one context, say $C[i]$ ($0\leq i \leq k$), can effectively transition to \textit{active}. The base case ($n=1$) ensures that upon $C[i]$'s transition, $C[k+1]$ cannot. Thus, precisely one context can transition to the \textit{active} state, and no other automaton can take a transition under the initial condition.
We conclude the proof of Lemma~\ref{lem:contextfirst}.

\end{proof}

\vspace{2mm}
\begin{lemma}
\label{lem:tcontext}
With any number $n$ of contexts, exactly one is \textit{active} at any time after the first activation, remaining so for at least $t_\texttt{context}$ time units.
\end{lemma}
\begin{proof}
Per Lemma~\ref{lem:contextfirst}, contexts initiate motion. Initially, any context instance can transition from \textit{idle}. Non-deterministically, one instance sets $cs\_enabled$ to false, blocking others in \textit{idle}. This active context resets $cs\_enabled$ to true after $t_{\texttt{context}}$, enabling another context activation. On exiting \textit{idle}, the next context synchronizes with the active one on $cs$, transitioning the former to \textit{idle}. 
\end{proof}

\vspace{2mm}
\begin{lemma}
\label{lem:contextswitch}
No context can be activated twice consecutively.
\end{lemma}

\begin{proof}
For a context $C[i]$ to transition back to \textit{idle}, it must synchronize on $cs?$, possible only when another context $C[j]$ ($i \neq j$) becomes active. Hence, $C[i]$ awaits at least one other context's activation before reactivation.
\end{proof}

The orchestrator, using a variable $cur$, tracks the active context. Upon a new context instance $x$ activation, the orchestrator:
\begin{enumerate}
\item Deactivates the current task automaton, $S[cur]$.
\item Activates the transition automaton from $cur$ to $x$, i.e., $T[cur][x]$.
\item Waits for transition completion via $done[cur][x]$ channel.
\item Activates the task automaton for context $x$, i.e., $S[x]$.
\end{enumerate}

\begin{lemma}
    \label{lem:transspecmutex}
    After the initial context activation, either a task automaton or a transition automaton is always in the \textit{active} state, but never both simultaneously.
    \end{lemma}
    \begin{proof}
    The orchestrator performs operations 1), 2), and 4) atomically, as modeled by committed locations in Figure~\ref{fig:automata}. Upon the first context activation, the orchestrator simultaneously activates a task automaton, ensuring an automaton's presence in the \textit{active} state. During a context switch, it deactivates the current task automaton, transitioning it to \textit{idle}, and activates a transition automaton, which becomes \textit{active}. Thus, at any moment, either a task automaton or a transition automaton is \textit{active}, but never both. Once the transition automaton transitions to \textit{idle}, the orchestrator activates the next task automaton in the same atomic operation, maintaining continuous activation of at least one automaton.
    \end{proof}

\vspace{2mm}

\begin{lemma}
    \label{lem:ttrans}
    A transition automaton's active period is bounded by $t_{\texttt{trans}}$ time units.
    \end{lemma}
    \begin{proof}
    The \textit{active} state of a transition automaton carries an invariant that necessitates a transition to \textit{idle} when $t^{t} \geq t_{\texttt{trans}}$. This time constraint ensures that the active period does not exceed $t_{\texttt{trans}}$ units.
    \end{proof}
    
\begin{theorem}
\label{thm:orchestration}
Given $n$ contexts, $n$ task automata, $n^2-n$ transition automata, and assuming $t_{\texttt{trans}} < t_{\texttt{context}}$, the following holds:
\begin{itemize}
\item Whenever a context $C[i]$ is active, $S[i]$ becomes active within $t_{\texttt{trans}}$ time units, and no other transition or task automaton becomes active.
\end{itemize}
\end{theorem}

\begin{proof}
\emph{Initial phase:} With the orchestrator and all and automata in \textit{idle}, Lemmas~\ref{lem:contextfirst} and \ref{lem:transspecmutex} imply the system is initially waiting for a single context to activate. If $C[i]$ becomes active, the orchestrator promptly activates $S[i]$ via the $go$ channel.

\emph{Context Switch:} Assuming a current active context $C[j]$ ($j \neq i$), Lemma~\ref{lem:tcontext} dictates $C[i]$ cannot activate within $t_\texttt{context}$ units. If $C[i]$ activates after $t_\texttt{context}$ units from $C[j]$'s activation, Lemmas~\ref{lem:ttrans} and \ref{lem:transspecmutex}, combined with $t_{\texttt{trans}} < t_{\texttt{context}}$, mean that $S[j]$ is active when $C[i]$ activates. The orchestrator then stops $S[j]$ and starts $T[j][i]$. Within $t_{\texttt{trans}}$ units, $S[i]$ activates. Given $t_{\texttt{trans}} < t_{\texttt{context}}$ and Lemmas~\ref{lem:tcontext} and \ref{lem:contextswitch}, it is guaranteed that $S[i]$ activates.
\end{proof}

\subsubsection{Starvation}
The orchestration system, while ensuring mutual exclusion between contexts and maintaining each active for at least $t_{\texttt{context}}$ units, does not dictate context scheduling. Thus, potential exists for some contexts to never activate, consequently preventing associated task automata from activation.

\begin{remark}
Assuming a fair scheduler activating all contexts  infinitely often and $t_{\texttt{trans}} < t_{\texttt{context}}$, Theorem \ref{thm:orchestration} implies all task automata will activate  infinitely often, preventing starvation of any context.
\end{remark}

\subsubsection{Model Checking}
We employed UPPAAL\footnote{UPPAAL model available: \url{rebrand.ly/cromeuppaalmodel}}\cite{behrmann2006uppaal} to model our orchestration system as a network of timed automata (Figure~\ref{fig:automata}), parametrized by context count $N$. Our model successfully verified the following properties:

\begin{center}
\textit{`No more than one T-Automaton is active at any time'}\\

    \texttt{A[] forall (i: N) forall (j: N) forall (l: N) forall (m: N) T\_automaton(i,l).Active \& T\_automaton(j,m).Active imply (i==j \& l==m)}
\end{center}

\begin{center}
\textit{`Always at most one context active at a time'}\\

    \texttt{A[] forall (i: N) forall (j: N) Context(i).Active \& Context(j).Active imply (i==j)}
\end{center}

\begin{center}
\textit{`There is always at most one S-automaton active at any given time'}\\

    \texttt{A[] forall (i: N) forall (j: N) S-automaton(i).Active \& S-automaton(j).Active imply (i==j)}
\end{center}

\begin{center}
\textit{`If any context becomes active its S-automaton will eventually become active'}\\
    \texttt{Context(i).Active --> S\_automaton(i).Active}
\end{center}

\begin{center}
\textit{`The system is deadlock-free'}\\
    \texttt{A[] not deadlock}
\end{center}

\begin{theorem}
    \label{thm:missionsatisfaction}[Mission Satisfaction]
    Given a mission $\mathcal{M} = \{X, G\}$, where each mission scenario $\mathcal{M}_i = (x_i, \gamma_i)$ corresponds to a context $x_i \in X$ and a LTL contract $\gamma_i$, and the composition of all controllers in the set $\Lambda^*$ as a finite-state machine $\Lambda^{||}$, assuming:
    \begin{enumerate}
    \item Each $\gamma_i$ is realized by a corresponding task controller $\lambda_i$ within $\Lambda^*$;
    \item For every possible context switch, a transition controller $k$ exists within $\Lambda^*$, and $t_{\texttt{trans}} < t_{\texttt{context}}$;
    \end{enumerate}
    Then the behaviors produced by $\Lambda^{||}$, combined with \CROME's orchestration system, will always satisfy the mission $\mathcal{M}$.
    \end{theorem}
    
    \begin{proof}
    Theorem~\ref{thm:orchestration} ensures that when a context $x_i$ becomes active, the corresponding task automaton activates within $t_\texttt{trans}$ units. The task and transition automata behaviors, representing the actual task and transition controllers, follow the \emph{activate()}, \emph{deactivate()}, \emph{trans\_start()}, and \emph{trans\_end()} functions. When these functions are not activated, the Mealy machines respond in synchronization with the timed automata ticks.
    
    When a context $x_i$ becomes active at time $t$, it implies that the mission scenario $\mathcal{M}_i = (x_i, \gamma_i)$ is active, with $\lambda_i \models \gamma_i$. The behavior of $\lambda_i'$ is equivalent to $\lambda_i$ when it is active. The \textit{Active} state in Section~\ref{rem:active} indicates that a task automaton is in the active state.
    
    Given that $\lambda_i$ satisfies $\gamma_i$ by construction, the mission $\mathcal{M}$ is always satisfied by the robot $\mathcal{R}$ as per the mission definition. Therefore, for any activated context $x_i$, the behaviors produced by $\Lambda^{||}$ in conjunction with \CROME's orchestration system always satisfy the mission scenario $\mathcal{M}_i$. Since this theorem holds for all $x_i \in X$, it holds for the entire mission $\mathcal{M}$, thereby proving the theorem.
    \end{proof}

\section{Case Study: Implementing Contextual Mission with \NAME}
\label{sec:casestudy}

In this section, we illustrate how a designer can use \NAME to model the goals of our running example and produce traces that always satisfy the contextual mission.

\subsection{Orchestration}

\begin{figure*}[h]
  \centering
  \includegraphics[width=0.8\linewidth]{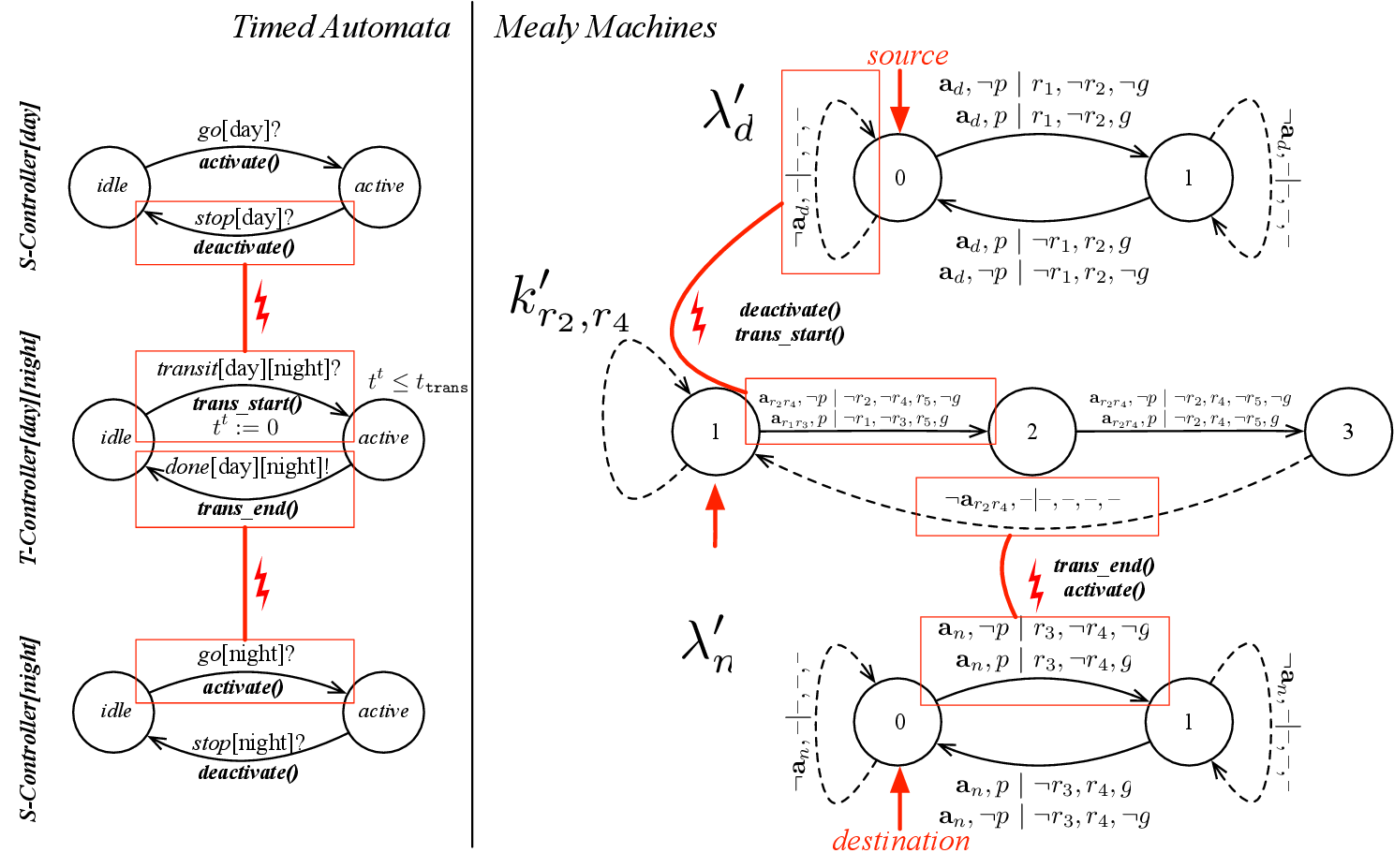}
  \caption{Synchronization between Timed Automata and Mealy Machines. The red arrows indicate the state of the machines at time $t$, when we have a context switch from context $day$ to $night$. The connected red boxes indicate transitions that happen in parallel.} 
  \label{fig:Mealysynch}
\end{figure*}

\vspace{2mm}
Consider the two contexts in our running example: $day$ and $night$. We have two instances of \textit{S-Automaton}, namely \textit{S[day]} and \textit{S[night]}, and two task controllers $\lambda_d'$ and $\lambda_n'$ that realize the specifications $\gamma_d$ and $\gamma_n$, which correspond to the contexts $day$ and $night$ respectively. Initially, we focus on goals $\mathcal{G}_1$, $\mathcal{G}_2$, and $\mathcal{G}_4$, omitting $\mathcal{G}_3$ which requires the robot to register a person. 
Figure~\ref{fig:Mealysynch} illustrates the Mealy machines $\lambda_d'$ and $\lambda_n'$ which realize $\gamma_d$ and $\gamma_n$. These machines have been adapted to accept the activation inputs $\textbf{a}_d$ and $\textbf{a}_n$.

At any time, only one activation input of $\Lambda^{||}$ is set to \textit{true}, indicating that only one Mealy machine in the composition $\Lambda^{||}$ is active. The following functions are executed on the transitions of the \textit{S-Automaton} and \textit{T-Automaton}:

\begin{itemize}
  \item \textbf{activate()}: Executed by every \textit{S-Automaton} instance when transitioning from \textit{idle} to \textit{active}. It activates the corresponding task controller. For example, \textbf{activate()} of \textit{S[day]} sets $\textbf{a}_d$ of $\lambda_d'$ to \textit{true}, and \textit{S[night]}'s \textbf{activate()} sets $\textbf{a}_n$ of $\lambda_n'$ to \textit{true}.
\item \textbf{deactivate()}: The opposite of \textbf{activate()}, it sets the activation inputs of the corresponding task controllers to \textit{false} after a context switch.
\item \textbf{trans\_start()}: Activates the appropriate transition controller for each context switch. The selected transition controller depends on the robot's location at the time of the context switch and its required destination in the new context.
\item \textbf{trans\_end()}: Executed when the transition controller reaches its final state, setting its activation input to \textit{false}.
\end{itemize}

\begin{figure*}[t]
  \centering
  \includegraphics[width=1\linewidth]{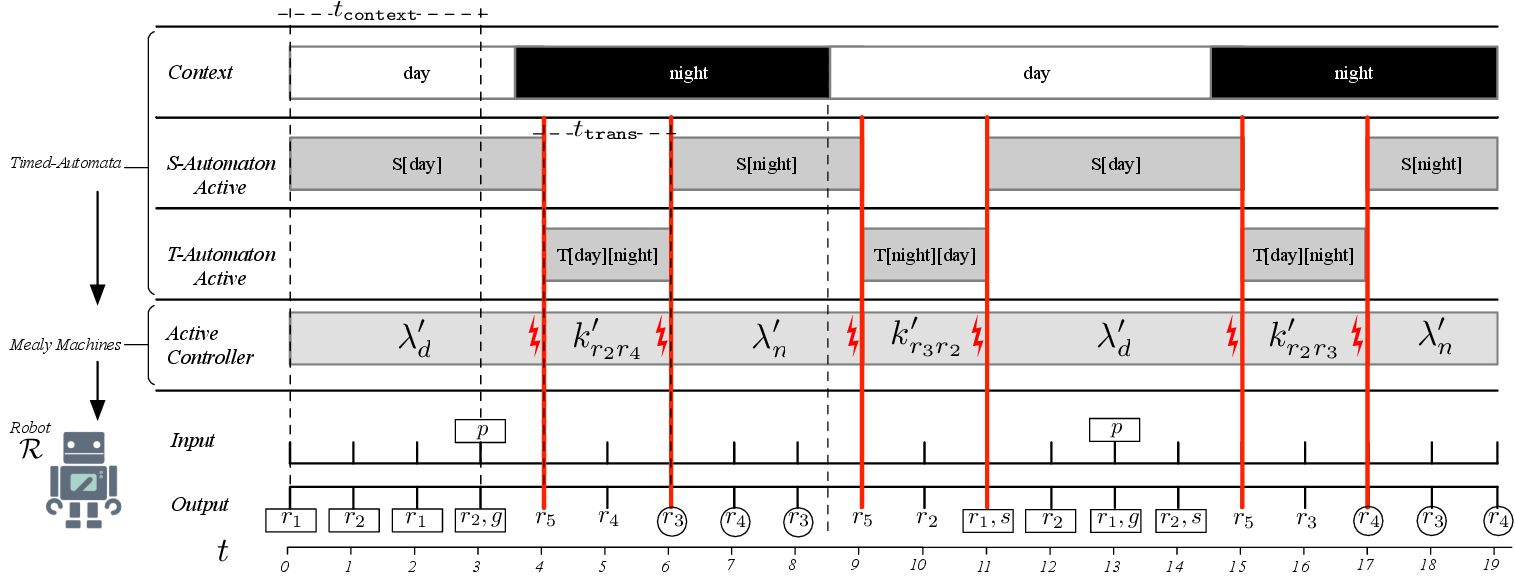}
  \caption{Timeline.} 
  \label{fig:timeline}
\end{figure*}

Consider a context switch at time $t$ where $\lambda_n'$ is active in state $0$, and $\lambda_d'$ is inactive in state $0$. Each state represents the last region reached by the robot, i.e., region $r_2$ for $\lambda_d'$ and region $r_4$ for $\lambda_n'$.

\textbf{Time t}: \textit{S[day]} transitions from \textit{active} to \textit{idle}, and \textit{T[day][night]} from \textit{idle} to \textit{active}. \textbf{deactivate()} sets $\textbf{a}_d$ to \textit{false}, and \textbf{trans\_start()} activates the transition controller connecting $\lambda_n'$ state $0$ to $\lambda_d'$ state $0$.

\textbf{Time t+1}: Both \textit{S[day]} and \textit{S[night]} are in \textit{idle}, while \textit{T[day][night]} is \textit{active}. The transition controller $k'_{r_2r_4}$ moves from state $2$ to $3$.

\textbf{Time t+2}: Transition controller $k'_{r_2r_4}$ reaches its final state. \textit{T[day][night]} transitions to \textit{idle}, executing \textbf{trans\_end()} and setting $\textbf{a}_{r_2r_4}$ to \textit{false}. Simultaneously, \textit{S[night]} transitions to \textit{active}, and \textbf{activate()} sets $\textbf{a}_n$ to \textit{true}.

In the comprehensive example, we include all goals $\mathcal{G}_1$, $\mathcal{G}_2$, $\mathcal{G}_3$, and $\mathcal{G}_4$ in the CGG under two scenarios: \textit{1)} $\mathcal{G}_d$ for daytime tasks of greeting and registering people in regions $r_1$ and $r_2$; \textit{2)} $\mathcal{G}_n$ for nighttime tasks of greeting people in regions $r_3$ and $r_4$.

Figure~\ref{fig:timeline} shows the event timeline, where the context switches from $day$ to $night$ after a minimum of $t_{\text{context}}$ time units. Timed automata synchronize after each context switch, with \textit{T-Automata} active for up to $t_{\text{trans}}$ time units. The trace of robot $\mathcal{R}$ is defined by its input-output relationship. The output symbols in \textit{boxes} are part of the trace projections for the contract $\gamma_d$ (related to goal $\mathcal{G}_d$), while those in \textit{circles} are for $\gamma_n$ (related to goal $\mathcal{G}_n$). At step 3, the robot perceives a person $p$ and immediately greets $g$, but registration $s$ occurs at step 11, the next time $day$ context is active and the transition period has ended.

\section{Implementation}
\label{sec:implementation}

\NAME is implemented as a web-based tool, available open-source\footnote{\NAME tool: \url{rebrand.ly/crometool}}. The tool leverages NuXMV~\cite{nuxmv} and STRIX~\cite{strix} as back-end engines for checking the satisfiability and realizability of specifications, respectively. After constructing the CGG, each mission scenario is realized as a task controller in the form of a Mealy machine.

For the realization of the transition controller, two main factors are considered:

1.\textit{Number of Location Pairs}: Depending on the mission specifications, certain context switches may necessitate the robot moving to different locations. In the worst-case scenario, where every context switch requires a location change, we define a set of regions $R$ as the union of mutually exclusive region sets $\Omega_{R1}, \Omega_{R2}, ..., \Omega_{Rn}$, corresponding to context variables $x_1, x_2, ..., x_n$. Thus, the total number of location pairs requiring a transition controller is $(|\Omega_{R1}| \times |\Omega_{R2}| \times ... \times |\Omega_{Rn}|) \times 2$, considering all possible pairs and permutations.

2. \textit{Generation of Transition Controllers}: Given the environment map, creating a transition controller essentially becomes a shortest path problem. Although \NAME uses reactive synthesis, any algorithm that solves this problem, such as Dijkstra's algorithm, could be utilized.

In the running example, up to eight transition controllers might be needed for transitions between the "day" and "night" contexts. However, \NAME employs an optimization algorithm considering the environment map and adjacent locations, minimizing the number of controllers. For instance, if a central location like $r_5$ in Figure~\ref{fig:illustrative} is adjacent to all other locations, the robot can reach any destination via $r_5$, reducing the number of necessary controllers.

\NAME's algorithm for generating transition controllers seeks a path from the source to the destination location within the time constraint of $t_{\texttt{trans}}$. It explores paths of length $1 < N < t_{\texttt{trans}}$ until a suitable one is found. A more sophisticated implementation in \NAME considers the environmental map and adjacency of locations to further reduce the number of controllers.

\section{Evaluation}
\label{sec:evaluation}

\newcommand\mstime{$35.6~$}

\newcommand\msstates{$119~$}
\newcommand\mstransitions{$357~$}
\newcommand\cstime{$0.14~$}
\newcommand\csstates{$5 ~\text{and}~ 3~$}
\newcommand\cstransitions{$~10 ~\text{and}~ 6~$}

\newcommand\nttrans{$16~$}
\newcommand\talltrans{$0.92~$}
\newcommand\cggtime{$0.54~$}
\newcommand\totaltime{$1.60~$}

\newcommand\cfaster{$254~$ times}
\newcommand\cfastermod{$22~$ times}

\begin{table}[t]
{\small
    \begin{TAB}(r,1cm,1cm)[5pt]{|r|c|c|}{|c|c|c|c|}
        & Monolithic LTL  & \NAME \\
        Synthesis Time (sec) & \mstime & \cstime \\
        Number of States & \msstates & \csstates \\
        Number of Transitions & \mstransitions & \cstransitions
    \end{TAB}    
    \caption{Comparison of \NAME with the synthesis from a single LTL formula in order to realize the running example.}
    \label{tab:evaluation}
    }
\end{table}

To evaluate our approach, we measured the synthesis time required by \NAME to generate traces for the robotic mission in our example. Given the manageable size of our example, we created a 'monolithic' LTL specification encompassing all desired behaviors and constraints. This specification required additional variables for context tracking, active signal management, and transition conditions. An extra Boolean variable was also necessary to manage potential halts in robot execution due to context switches, significantly enlarging the LTL formula. The complete specification and test can be found online\footnote{\NAME evaluation: \url{rebrand.ly/cromeevaluation}}.

Table~\ref{tab:evaluation} compares the synthesis time and size of the controller for the monolithic LTL approach versus \NAME. The monolithic LTL synthesis took \mstime seconds, yielding a controller with \msstates states and \mstransitions transitions. In contrast, \NAME synthesized the smaller specifications $\gamma_1$ and $\gamma_2$ related to the goals $\mathcal{G}{ng}$ and $\mathcal{G}{dgr}$ in just \cstime seconds, which is \cfaster faster than the monolithic approach. The resulting controllers for these two contexts had \csstates states and \cstransitions transitions, respectively.

The table does not include the time for generating transition controllers. \NAME's reactive synthesis approach took \talltrans seconds to create all 16 transition controllers. Additionally, the CGG creation, involving 11 satisfiability checks and 10 validity checks to derive $\gamma_1$ and $\gamma_2$ from the four input goals, took \cggtime seconds.

Overall, including CGG creation and transition controller synthesis, the total time was \totaltime seconds, still \cfastermod faster than synthesizing the monolithic LTL specification.

\section{Related Work}\label{sec:related}

Despite valuable contributions to the state-of-the-art in the field, dealing with the variability of the environment and addressing context switching remain open challenges. This section presents works related to mission specifications (Section~\ref{sec:missionSpecification}) and reactive synthesis applied to robotics (Section~\ref{sec:reactiveSynthesis}). We also briefly discuss the concept of context (Section~\ref{sec:bncontext}).

\subsection{Mission Specification}\label{sec:missionSpecification}

A robotic mission is typically defined as a set of tasks or goals for a robot to accomplish. Various approaches for specifying missions have been proposed, including logics~\cite{menghi2018multi,ulusoy2011optimal,fainekos2009temporal,guo2013revising,wolff2013automaton,kress2011robot,doi:10.1177/0278364914546174}, state charts~\cite{bohren2010smach,thomas2013new,klotzbucher2012coordinating}, and Petri Nets~\cite{wang1991petri,ziparo2008petri}. However, the need for formal methods knowledge can be a barrier to wider adoption. Domain-specific languages (DSLs) have been introduced to make mission specification more accessible~\cite{GLRSWWCA12,campusano.ea:2017:live,DBLP:journals/corr/SchwartzNAM14,Ruscio2016,Bozhinoski2015,Ciccozzi4496,Doherty2012}. These DSLs provide high-level, user-friendly interfaces for mission specification.

A recent survey by Dragule et al.~\cite{dragule2021b} on DSLs for robot mission specifications classifies them into internal and external DSLs and overviews their tooling support. Additionally, robotic patterns have been proposed to address recurrent mission specification problems~\cite{SpecificationPatternsTSE}. The PsALM tool by Menghi et al.~\cite{PSALM} enables the specification of robotic missions using mission specification patterns, which capture robot movements and actions as the robot moves in the environment. \NAME supports 22 patterns~\cite{SpecificationPatternsTSE} covering various mission requirements, organized into core movement patterns, triggers, and avoidance patterns. For instance, the \textit{Strict Ordered Patrolling} pattern\footnote{\url{http://roboticpatterns.com/pattern/strictorderedpatrolling/}} can be used to specify a requirement for a robot to patrol a set of locations in a specific order. Let $l_1, l_2$, and $l_3$ be the atomic propositions of type \textit{location} that the robot must visit in the given order. The mission requirement can be reformulated in LTL as shown above, demonstrating how a robotic pattern can simplify the complex task of mission specification.  

Despite these efforts, formalizing mission requirements in specifications remains challenging. \NAME uses a catalog of patterns by Dywer et al.~\cite{patterns, dywerpatterslink} to formally yet intuitively express the robot's tasks.

\subsection{Context}\label{sec:bncontext}

The concept of context has been extensively discussed in the literature. In ubiquitous computing, context may include location, time of day, nearby objects and people, and temperature as proposed by Krumm ~\cite{Krumm2010}. Dey and Abowd~\cite{dey2001understanding} define context as any information characterizing an entity's situation. In robotics, Bloisi et al.~\cite{Bloisi2016} describe mission-related contexts as choices enabling robots to adapt to different situations.

Capturing the exceptional behaviors robots must exhibit to cope with real-world environmental variability is a significant challenge in specifying robotic missions (Garcia et al.~\cite{ESECFSE2020)}. \NAME formalizes context as a property associated with a goal and expressed through logic formulas, building on existing definitions. This approach enables the specification of exceptional behaviors in response to environmental changes.

\subsection{Reactive Synthesis}\label{sec:reactiveSynthesis}

Reactive synthesis in robotics focuses on synthesizing controllers. A compositional approach to reactive synthesis by He et al.~\cite{Moshe2019} uses linear temporal logic on finite traces (LTLf) for finite-horizon tasks. The work by Salar~\cite{Salar2018} presents algorithms for synthesizing controllers for swarm robotic systems. Maoz and Ringert in~\cite{MaozRose} identify challenges in applying reactive synthesis to robotics, including synthesis algorithms, development processes and tools, declarative specification writing, and data and time abstraction.

One method to handle real-world variability, as discussed in Section~\ref{sec:missionSpecification}, is to decompose mission specifications into context-dependent sub-missions. The work by Nahabedian et al.~\cite{nahabedian_2020} proposes automatic controller computation for transitioning between specifications. However, this runtime synthesis may be problematic for immediate context switching. The work by D'Ippolito et al.~\cite{ippolito2014_hope} introduces a hierarchy of discrete event controllers for graceful degradation and progressive enhancement, enabling instantaneous switching with behavioral guarantees, though it requires predefined hierarchical layers.

\NAME offers an end-to-end solution for specifying and realizing contextual mission specifications. It simplifies mission specification by decomposing it into context-dependent sub-missions. After specifying the mission in sub-missions, \NAME automatically synthesizes (i) a controller for each sub-mission satisfying the contextual mission and (ii) an orchestrator to support context switching, ensuring the overall contextual mission's satisfaction. Unlike the approach in~\cite{nahabedian_2020}, \NAME allows indefinite context switching while orchestrating the controllers.

\section{Conclusions}
\label{sec:conclusion}

In this paper, we introduced a comprehensive design framework tailored for formalizing and realizing robotic requirements in contextual missions. Central to our approach is \NAME, a novel tool that bridges the gap between informal requirements and formal mission specifications. By leveraging specification patterns and context-based modeling, \NAME translates ambiguous requirements into precise, formal goals. These goals are structured into assume-guarantee contracts and systematically organized within a Contract-Based Goal Graph (CGG).

A key innovation of our framework lies in its ability to handle dynamic controller switching in response to varying mission scenarios, each dictated by distinct, mutually exclusive contexts. The orchestration of these context-specific controllers is adeptly managed through \CROME, which not only analyzes and synthesizes the mission specification but also ensures efficient and seamless transitions between controllers.

Our empirical results demonstrate a significant improvement in synthesis efficiency, with \NAME yielding leaner, more modular specifications that accelerate the synthesis process. This advance represents a substantial stride in the realm of robotic mission planning, particularly in scenarios demanding adaptability and precision.

Looking ahead, we aim to broaden the scope of our framework. Our future endeavors include the development of a new logic language specifically tailored to express and handle the complexities inherent in the class of problems addressed in this study. Additionally, we plan to undertake more extensive validation of our approach, particularly through larger-scale industrial case studies. This will not only test the robustness of our framework but also provide valuable insights for further refinement and adaptation to real-world applications.

\section*{Acknowledgment}

This research was supported in part by the Wallenberg AI Autonomous Systems and Software Program (WASP), funded by the Knut and Alice Wallenberg Foundation.

The authors acknowledge the support of the PNRR MUR project VITALITY (ECS00000041), Spoke 2 ASTRA - ``Advanced Space Technologies and Research Alliance", of the PNRR MUR project CHANGES (PE0000020), Spoke 5 ``Science and Technologies for Sustainable Diagnostics of Cultural Heritage'', the PRIN project P2022RSW5W -
RoboChor: Robot Choreography, the PRIN project 2022JKA4SL - HALO: etHical-aware AdjustabLe autOnomous systems,
and of the MUR (Italy) Department of Excellence 2023 - 2027 for GSSI. 

The work of P.\ Pelliccione was also partially supported by the Centre of EXcellence on Connected, Geo-Localized and Cybersecure Vehicles (EX-Emerge), funded by the Italian Government under CIPE resolution n. 70/2017 (Aug. 7, 2017).

G.\ Schneider was partially supported by the Swedish Research Council (Vetenskapsr\aa det) under
grant Nr.~2019-04951.

\ifCLASSOPTIONcaptionsoff
  \newpage
\fi

\bibliographystyle{IEEEtran}
\bibliography{mybibfile}

\end{document}